\documentclass[journal,final]{IEEEtran}

\usepackage{graphicx}
\usepackage{cite}
\usepackage[cmex10]{amsmath}
\usepackage{amsfonts}
\usepackage{amssymb}
\usepackage{subfigure}
\usepackage{color}
\usepackage{amsthm}
\usepackage{enumerate}
\usepackage{relsize}

\newcommand{\bm}{\mathbf} 
\newcommand{\be}{\begin{equation}}
\newcommand{\ee}{\end{equation}}
\newcommand{\bse}{\begin{subequations}}
\newcommand{\ese}{\end{subequations}}
\newcommand{\bea}{\begin{eqnarray}}
\newcommand{\eea}{\end{eqnarray}}

\newcommand{\ba}{{\bm a}}
\newcommand{\bb}{{\bm b}}

\newcommand{\e}{{\bm e}}

\newcommand{\bA}{{\bm A}}
\newcommand{\bK}{{\bm K}}

\newcommand{\bR}{{\bm R}}

\newcommand{\bF}{{\bf F}}
\newcommand{\bD}{{\bf D}}
\newcommand{\bC}{{\bf C}}

\newcommand{\bG}{{\bf G}}
\newcommand{\bH}{{\bf H}}

\newcommand{\bg}{{\bf g}}

\newcommand{\bh}{{\bf h}}
\newcommand{\bp}{{\bm p}}

\newcommand{\bw}{{\bf w}}
\newcommand{\bd}{{\bf d}}

\newcommand{\by}{{\bf y}}

\newcommand{\bzero}{{\bf 0}}

\newcommand{\eye}{{\bm I}}
\newcommand{\I}{{\bm I }}

\newcommand{\BD}{{\boldsymbol{\mathcal D}}}
\newcommand{\BW}{{\boldsymbol{\mathcal W}}}
\newcommand{\BP}{{\boldsymbol{\mathcal P}}}
\newcommand{\BE}{{\mathbb E}}

\newcommand{\bGamma}{\mbox{\boldmath$\Gamma$}}

\newcommand{\bPsi}{\mbox{\boldmath$\Psi$}}
\newcommand{\bpsi}{\mbox{\boldmath$\psi$}}
\newcommand{\bnu}{\mbox{\boldmath$\nu$}}
\newcommand{\bgamma}{\mbox{\boldmath$\gamma$}}
\newcommand{\bmu}{\mbox{\boldmath$\mu$}}

\newcommand{\sinc}{\mbox{sinc}}

\theoremstyle{definition}
\newtheorem{proposition}{Proposition}

\graphicspath{{figures/}} 

\title{Filter Bank Multicarrier in Massive MIMO: Analysis and Channel Equalization}

\author{\normalsize Amir Aminjavaheri, Arman Farhang, and Behrouz Farhang-Boroujeny 

\thanks{This publication has emanated from research supported in part by a research grant from Science Foundation Ireland (SFI) and is co-funded under the European Regional Development Fund under Grant Number 13/RC/2077. Parts of the concepts based on which the contents of this paper are built have been presented in \cite{confversion}.}

\thanks{A.~Aminjavaheri and B.~Farhang-Boroujeny are with the Electrical and Computer Engineering Department, University of Utah, Salt Lake City, USA (e-mail: \{aminjav, farhang\}@ece.utah.edu).}

\thanks{A.~Farhang is with the School of Electrical and Electronic Engineering, University College Dublin, Ireland, Dublin4 (e-mail: arman.farhang@ucd.ie).} \vspace{-0.3in}}

\begin{document}

\maketitle

\begin{abstract}
We perform an asymptotic study of the performance of filter bank multicarrier (FBMC) in the context of massive multi-input multi-output (MIMO). We show that the  effects of channel distortions, i.e., intersymbol interference and intercarrier interference, do not vanish as the base station (BS) array size increases. As a result, the signal-to-interference-plus-noise ratio (SINR) cannot grow unboundedly by increasing the number of BS antennas, and is upper bounded by a certain deterministic value. We show that this phenomenon is a result of the correlation between the multi-antenna combining tap values and the channel impulse responses between the mobile terminals and the BS antennas. To resolve this problem, we introduce an efficient equalization method that removes this correlation, enabling us to achieve arbitrarily large SINR values by increasing the number of BS antennas. We perform a thorough analysis of the proposed system and find analytical expressions for both equalizer coefficients and the respective SINR. 
\end{abstract}


\section{Introduction}

\IEEEPARstart{M}{assive} multiple-input multiple-output (MIMO) is one of the primary technologies currently considered for the next generation of wireless networks, \cite{boccardi2014five}. In a massive MIMO system, the base station (BS) is equipped with a large number of antenna elements, in the order of hundreds or more, and is simultaneously serving tens of mobile terminals (MTs). By coherent processing of the signals over the BS antennas, the effects of uncorrelated noise and multiuser interference can be made arbitrarily small as the BS array size increases, \cite{marzetta2010noncooperative,rusek2013scaling}. Hence, unprecedented network capacities can be achieved.

Filter bank multicarrier (FBMC) is a waveform that has gained an increased attention in the recent years due to its improved spectral properties compared to orthogonal frequency division multiplexing (OFDM), \cite{farhang2011ofdm,perez2015mimo,aminjavaheri2015impact}. The application of FBMC to massive MIMO channels has been recently studied in \cite{armanfarhang2014filter}, where its so-called self-equalization property leading to a channel flattening effect was reported through simulations. According to this property, the effects of channel distortions (i.e., intersymbol interference and intercarrier interference) will diminish by increasing the number of BS antennas. The authors in \cite{rottenberg2017performance} obtain the asymptotic mean squared error (MSE) performance of FBMC in massive MIMO channels. Their analysis shows that the MSE becomes uniform across different subcarriers as a result of the channel hardening effect. In \cite{aminjavaheri2015frequency}, multi-tap equalization per subcarrier is proposed for FBMC-based massive MIMO systems to improve the equalization accuracy as compared to the single-tap equalization at the expense of a higher computational complexity. The authors in \cite{farhang2014pilot} show that the pilot contamination problem in multi-cellular massive MIMO networks, \cite{marzetta2010noncooperative}, can be resolved in a straightforward manner with FBMC signaling due to its special structure. These studies prove that FBMC is an appropriate match for massive MIMO and vice versa as they can both bring pivotal properties into the picture of the next generations of wireless systems. Specifically, this combination is of a great importance as not only the same spectrum is being simultaneously utilized by all the users but it is also used in a more efficient manner compared to OFDM.

Since the literature on FBMC-based massive MIMO is not mature yet, these systems need to go through meticulous analysis and investigation. In particular, in this paper, we perform an in-depth analysis on the performance of FBMC in massive MIMO channels. The focus of this paper is on the uplink transmission, while the theories and proposed techniques are trivially applicable to the downlink as well. We consider single-tap equalization per subcarrier, and investigate the performance of three most prominent linear combiners, namely, maximum-ratio combining (MRC), zero-forcing (ZF), and minimum mean-square error (MMSE). We show that the self-equalization property shown through simulations and claimed in \cite{armanfarhang2014filter} and \cite{aminjavaheri2015frequency} is not very accurate. More specifically, by increasing the number of BS antennas, the channel distortions average out only up to a certain extent, but not completely. Thus, the SINR saturates at a certain deterministic level. This determines an upper bound for the SINR performance of the system.

Our main contributions in this paper are the following; (i) We derive an analytical expression for the SINR saturation level using MRC, ZF, and MMSE combiners. (ii) We propose an effective equalization method to resolve the saturation problem. With the proposed equalizer in place, SINR grows without a bound by increasing the BS array size, and arbitrarily large SINR values are achievable. (iii) An efficient implementation of the proposed equalization method through using some concepts from multi-rate signal processing is also presented. (iv) Finally, we perform a thorough analysis of the proposed system, and find the analytical expressions for the SINR in the cases of MRC and ZF detectors. All the above analyses are evaluated and confirmed through numerical simulations.

It is worth mentioning that although the theories developed in this paper are applicable to all types of FBMC systems, the formulations are based on the most common type in the literature that was developed by Saltzberg, \cite{saltzberg1967performance}, and is known by different names including OFDM with offset quadrature amplitude modulation (OFDM/OQAM), FBMC/OQAM, and staggered multitone (SMT), \cite{farhang2011ofdm}. Throughout this paper, we refer to it as FBMC for simplicity.

The rest of the paper is organized as follows. To pave the way for the derivations presented in the paper, we review the FBMC principles in Section \ref{sec:system_model}. In Section \ref{sec:mmimo_fbmc}, we present the asymptotic equivalent channel model between the mobile terminals and the BS in an FBMC massive MIMO setup. This analysis will lead to an upper bound for the SINR performance of the system. Our proposed equalization method is introduced in Section \ref{sec:prototype_modify}. In Section \ref{sec:self-equalization}, we study the FBMC in massive MIMO from a frequency-domain perspective, leading to some insightful remarks regarding these systems. In Section \ref{sec:sinr}, we find the SINR performance of the FBMC system incorporating the proposed equalization method. The mathematical analysis of the paper as well as the efficacy of the proposed filter design technique are numerically evaluated in Section \ref{sec:numerical_results}. Finally, we conclude the paper in Section \ref{sec:conclusion}.

\textit{Notations:} Matrices, vectors and scalar quantities are denoted by boldface uppercase, boldface lowercase and normal letters, respectively. $A^{m,n}$ represents the element in the $m^{\rm{th}}$ row and the $n^{\rm{th}}$ column of $\bA$ and $\bA^{-1}$ signifies the inverse of $\bA$. $\I_M$ is the identity matrix of size $M\times M$, and $\bD={\rm diag}\{\ba\}$ is a diagonal matrix whose diagonal elements are formed by the elements of the vector $\ba$. The superscripts $(\cdot)^{\rm T}$, $(\cdot)^{\rm H}$ and $(\cdot)^\ast$ indicate transpose, conjugate transpose, and conjugate operations, respectively. The linear convolution is denoted by $\star$. The real and imaginary parts of a complex number are denoted by $\Re\{\cdot\}$ and $\Im\{\cdot\}$, respectively. $\mathbb{E}\{\cdot\}$ denotes the expected value of a random variable, and ${\rm tr} \{ \cdot\}$ is the matrix trace operator. The notation $\mathcal{CN}(0,\sigma^2)$ represents the circularly-symmetric complex normal distribution with zero mean and variance $\sigma^2$. Finally, $\delta_{ij}$ represents the Kronecker delta function.

\section{FBMC Principles} \label{sec:system_model}

\begin{figure*}[!t]
\centering
\includegraphics[scale=0.78]{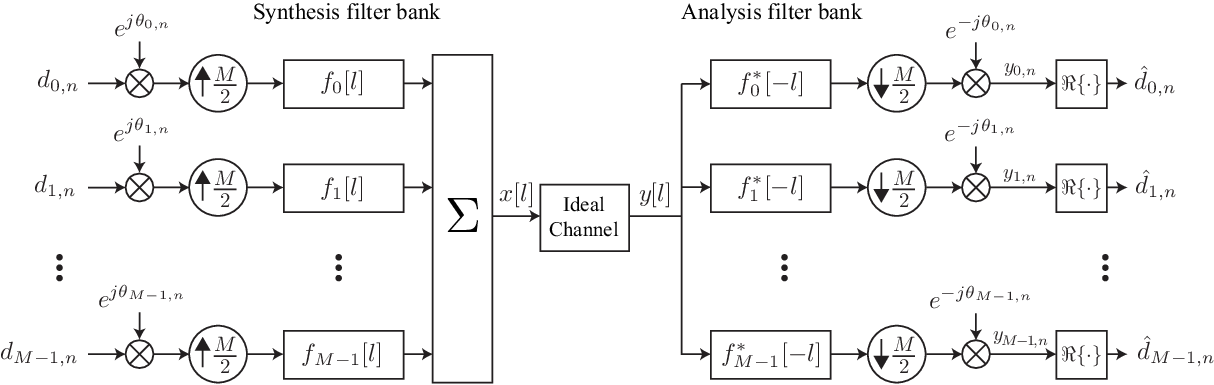}
\caption{Block diagram of the FBMC transceiver in discrete time.}
\label{fig:block_diagram}
\end{figure*}

\begin{figure*}[!t]
\centering
\includegraphics[scale=0.78]{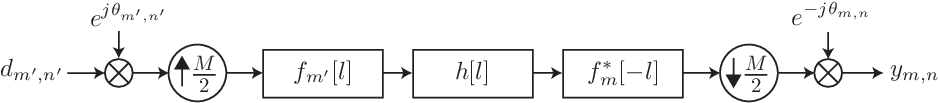}
\caption{The equivalent channel between the transmitted data symbol at time-frequency point $(m',n')$ and the demodulated symbol at time-frequency point $(m,n)$.}
\label{fig:equiv_chan}
\end{figure*}

We present the theory of FBMC in discrete time. Let $d_{m,n}$ denote the real-valued data symbol transmitted over the $m^{\rm th}$ subcarrier and the $n^{\rm th}$ symbol time index. The total number of subcarriers is assumed to be $M$. In order to avoid interference between the symbols and, thus, maintain the orthogonality, the data symbol $d_{m,n}$ is phase adjusted using the phase term $e^{j\theta_{m,n}}$, where $\theta_{m,n} = \frac{\pi}{2}(m+n)$. Accordingly, each symbol has a $\pm\frac{\pi}{2}$ phase difference with its adjacent neighbors in both time and frequency. The symbols are then pulse-shaped using a prototype filter $f[l]$, which has been designed such that $q[l] = f[l] \star f^*[-l]$ is a Nyquist pulse with zero crossings at $M$ sample intervals. The length of the prototype filter, $f[l]$, is usually expressed as $L_{\rm f} = \kappa M$, where $\kappa$ is called the overlapping factor\footnote{The overlapping factor indicates the number of adjacent FBMC symbols overlapping in the time domain.}. To express the above procedure in a mathematical form, the discrete-time FBMC waveform can be written as, \cite{farhang2014filter},
\bea \label{eqn:fbmc_waveform}
x[l] = \sum_{n=-\infty}^{+\infty} \sum_{m=0}^{M-1} d_{m,n} a_{m,n}[l],
\eea
where
\begin{align}
a_{m,n}[l] &= f_m[l - nM/2] e^{j\theta_{m,n}}.
\end{align}
Here, $f_m[l] \triangleq f[l] e^{j\frac{2\pi ml}{M}}$ is the prototype filter modulated to the center frequency of the  $m^{\rm th}$ subcarrier, and the functions $a_{m,n}[l]$, for $m \in \{0,\dots,M-1\}$ and $n\in\{-\infty,\dots,+\infty\}$, can be thought as a set of basis functions that are used to modulate the data symbols. Note that the spacing between successive symbols in the time domain is $M/2$ samples. In the frequency domain, the spacing between successive subcarriers is $1/M$ in normalized frequency scale. It can be shown that the basis functions $a_{m,n}[l]$ are orthogonal in the real domain, \cite{farhang2014filter}, i.e.,
\begin{align} \label{eqn:orthogonality}
\langle a_{m,n}[l], a_{m',n'}[l] \rangle_\Re &= \Re \bigg\{ \sum_{l = -\infty}^{+\infty} a_{m,n}[l] a^*_{m',n'}[l] \bigg\} \nonumber \\
&= \delta_{mm'} \delta_{nn'} .
\end{align}
As a result, the data symbols can be extracted from the synthesized signal, $x[l]$, according to
\be
d_{m,n} = \langle x[l], a_{m,n}[l] \rangle_\Re.
\ee
Fig.~\ref{fig:block_diagram} shows the block diagram of the FBMC transceiver. Note that considering the transmitter prototype filter $f[l]$, and the receiver prototype filter $f^*[-l]$, the overall effective pulse shape $q[l] = f[l] \star f^*[-l]$ is a Nyquist pulse by design. Also, in practice, in order to efficiently implement the synthesis (transmitter side) and analysis (receiver side) filter banks, one can incorporate the polyphase implementation to reduce the computational complexity, \cite{farhang2011ofdm}.

The presence of a frequency-selective channel leads to some distortion in the received signal. Thus, one may adopt some sort of equalization to retrieve the transmitted symbols at the receiver side. In this paper, we limit our study to a case where the channel impulse response remains time-invariant over the interval of interest. Accordingly, the received signal at the receiver can be expressed as 
\begin{align} \label{eqn:y}
y[l] &= h[l] \star x[l] + \nu[l],
\end{align}
where $h[l]$ represents the channel impulse response, and $\nu[l]$ is the additive white Gaussian noise (AWGN). We denote the length of the channel impulse response by $L_{\rm h}$.

At the receiver, after matched filtering and phase compensation, and before taking the real part (see Fig.~\ref{fig:block_diagram}), the demodulated signal $y_{m,n}$ can be expressed as
\be \label{eqn:demod_symbol}
y_{m,n} = \sum_{n'=-\infty}^{+\infty} \sum_{m'=0}^{M-1} H_{mm',nn'} \hspace{1pt} d_{m',n'} + \nu_{m,n} ,
\ee
where $\nu_{m,n}$ is the noise contribution, and the interference coefficient $H_{mm',nn'}$ can be calculated according to
\bse \label{eqn:siso_equiv_chan}
\begin{align} 
H_{mm',nn'} &=  h_{mm'}[n-n'] \hspace{1pt} e^{j(\theta_{m',n'}-\theta_{m,n})}, \\
h_{mm'}[n]  &= \Big( f_{m'}[l] \star h[l] \star f_m^\ast[-l] \Big)_{\downarrow \frac{M}{2}}.
\end{align}
\ese
The symbol $\downarrow \frac{M}{2}$ denotes $\frac{M}{2}$-fold decimation. In (\ref{eqn:siso_equiv_chan}), $h_{mm'}[n]$ is the equivalent channel impulse response between the transmitted symbols at subcarrier $m'$ and the received ones at subcarrier $m$. This includes the effects of the transmitter filtering, the multipath channel, and the receiver filtering; see Fig. \ref{fig:equiv_chan}. According to (\ref{eqn:demod_symbol}), the demodulated symbol $y_{m,n}$ suffers from interference originating from other time-frequency symbols. In practice, the prototype filter $f[l]$ is designed to be well localized in time and frequency. As a result, the interference is limited to a small neighborhood of time-frequency points around the desired point $(m,n)$.


In order to devise a simple equalizer to combat the frequency-selective effect of the channel, it is usually assumed that the symbol period $M/2$ is much larger than the channel length $L_{\rm h}$, or equivalently, the channel frequency response is approximately flat over each subcarrier band. With this assumption, the demodulated signal $y_{m,n}$ can be expressed as,  \cite{lele2008channel},
\be \label{eqn:siso_demod_symbol2}
y_{m,n} \approx H_{m} \big( d_{m,n} + u_{m,n} \big) + \nu_{m,n},
\ee
where $H_m \triangleq  \sum_{l=0}^{L_{\rm h}-1} h[l]  e^{-j\frac{2\pi ml}{M}}$ is the channel frequency response at the center of the $m^{\rm th}$ subcarrier. The term $u_{m,n}$ is called the \textit{intrinsic interference} and is purely imaginary. This term represents the contribution of the intersymbol interference (ISI) and intercarrier interference (ICI) from the adjacent time-frequency symbols around the desired point $(m,n)$. Based on  (\ref{eqn:siso_demod_symbol2}), the effect of channel distortions can be compensated using a single-tap equalizer per subcarrier. After equalization, what remains is the real-valued data symbol $d_{m,n}$, the imaginary term $u_{m,n}$, and the noise contribution. By taking the real part from the equalized symbol, one can remove the intrinsic interference and obtain an estimate of $d_{m,n}$.

It should be noted that the performance of the above single-tap equalization primarily depends on the validity of the assumption that the symbol duration is much larger than the channel length, or equivalently, the frequency response of the channel is approximately flat over the pass-band of each subcarrier. On the other hand, in highly frequency-selective channels, where the above assumption is not accurate any more, more advanced multi-tap equalization methods (see \cite{perez2015mimo,ihalainen2011channel}) should be deployed to counteract the multipath channel distortions. 

\section{Massive MIMO FBMC: Asymptotic Analysis} \label{sec:mmimo_fbmc}

In this section, we first extend the formulation of the previous section to massive MIMO channels. Then, we show that linear combining of the signals received at the BS antennas using the channel frequency coefficients leads to a residual interference that does not fade away even with an infinite number of BS antennas. Hence, we conclude, the SINR is upper bounded by a certain deterministic value, and arbitrarily large SINR performances cannot be achieved as the number of BS antennas grows. 

We consider a single-cell massive MIMO setup \cite{marzetta2010noncooperative}, with $K$ single-antenna MTs that are simultaneously communicating with a BS equipped with an array of $N$ antenna elements. As mentioned earlier, in this paper, we consider the uplink transmission while the results and our proposed technique are trivially applicable to the downlink transmission as well.

Let $x_k[l]$ represent the transmit signal of the terminal $k$. The received signal at the $i^{\rm th}$ BS antenna can be expressed as
\be \label{eqn:yi}
y_i[l] = \sum_{k=0}^{K-1} x_k[l] \star h_{i,k}[l] + \nu_i[l],
\ee
where $h_{i,k}[l]$ is the channel impulse response between the $k^{\rm th}$ terminal and the $i^{\rm th}$ BS antenna, and $\nu_i[l]$ is the additive noise at the input of the $i^{\rm th}$ BS antenna. We assume that the samples of the noise signal $\nu_i[l]$ are a set of independent and identically distributed (i.i.d.) $\mathcal{CN}(0,\sigma_\nu^2)$ random variables.

For a given terminal $k$, we model the corresponding channel responses using the channel power delay profile (PDP) $p_k[l], l = 0,\dots,L_{\rm h}-1$. In particular, we assume that the channel tap $h_{i,k}[l]$, $l \in \{0,\dots,L_{\rm h}-1\}$, follows a $\mathcal{CN}(0,p_k[l])$ distribution, and different taps are assumed to be independent. The above assumption implies that the BS antenna array is sufficiently compact so that the channel responses corresponding to a particular user and different BS antennas are subject to the same channel PDP. We also assume that the channels corresponding to different terminals and different BS antennas are independent. 
Moreover, for each terminal, the average transmitted power is assumed to be equal to one, i.e., $\mathbb{E}\{ |x_k[l]|^2\} = 1$. To simplify the analysis throughout the paper, we assume that the BS has a perfect knowledge of the channel state information (CSI).

Following (\ref{eqn:yi}), we can extend (\ref{eqn:demod_symbol}) to the MIMO case according to
\be \label{eqn:mimo_demod_symbol1}
\by_{m,n} = \sum_{n'=-\infty}^{+\infty} \sum_{m'=0}^{M-1} \bH_{mm',nn'} \bd_{m',n'} + \bnu_{m,n},
\ee
where $\by_{m,n}$ is an $N \times 1$ vector containing the demodulated symbols corresponding to different BS antennas, $\bd_{m,n}$ is a $K \times 1$ vector containing the real-valued data symbols of all the $K$ terminals transmitted at the $m^{\rm th}$ subcarrier and the $n^{\rm th}$ time instant, $\bnu_{m,n}$ is the noise contribution across different BS antennas, and $\bH_{mm',nn'}$ is an $N\times K$ channel matrix. The element $(i,k)$ of $\bH_{mm',nn'}$ can be calculated according to
\bse \label{eqn:mimo_equiv_chan}
\begin{align} 
H_{mm',nn'}^{i,k} &=  h_{mm'}^{i,k}[n-n'] \hspace{1pt} e^{j(\theta_{m',n'}-\theta_{m,n})}, \\
h_{mm'}^{i,k}[n]  &= \Big(f_{m'}[l] \star h_{i,k}[l] \star f_m^\ast[-l]\Big)_{\downarrow \frac{M}{2}}.
\end{align}
\ese

We assume that the BS uses a single-tap equalizer per antenna per subcarrier. Accordingly, combining the elements of $\by_{m,n}$ using an $N \times K$ matrix $\BW_m$, and taking the real part from the resulting signal, the estimate of the transmitted data symbols for all the terminals can be obtained as
\begin{align} \label{eqn:mimo_est_d1}
\hat{\bd}_{m,n} &= \Re \left\{ \BW_m^{\rm H} \hspace{2pt} \by_{m,n} \right\} \nonumber \\
&= \Re \Big\{ \sum_{n'=-\infty}^{+\infty} \sum_{m'=0}^{M-1} \BW_m^{\rm H} \bH_{mm',nn'} \bd_{m',n'}  + \BW_m^{\rm H} \bnu_{m,n} \Big\} \nonumber \\
&= \Re \Big\{  \sum_{n'=-\infty}^{+\infty} \sum_{m'=0}^{M-1}  \bG_{mm',nn'} \bd_{m',n'}  + \bnu'_{m,n} \Big\} ,
\end{align}
where $\bG_{mm',nn'} \triangleq \BW_m^{\rm H} \bH_{mm',nn'}$, and $\bnu'_{m,n} \triangleq \BW_m^{\rm H} \bnu_{m,n}$. Here, we examine MRC, ZF, and MMSE linear combiners. These combiners can be formed as 
\be \label{eqn:mrc_zf_mmse}
\BW_m = 
\begin{cases}
	\bH_m \bD_m^{-1} ,  & {\rm for~~ MRC}, \\
	\bH_m \left( \bH_m^{\rm H} \bH_m \right)^{-1},	   & {\rm for~~ ZF}, \\
	\bH_m \left( \bH_m^{\rm H} \bH_m + \sigma_\nu^2 \eye_K \right)^{-1},	   & {\rm for~~ MMSE}, 
\end{cases}
\ee
where $\bH_m$ is the channel coefficient matrix at the center of the $m^{\rm th}$ subcarrier, i.e., $H_m^{i,k} \triangleq \sum_{l=0}^{L_{\rm h}-1} h_{i,k}[l] e^{-j \frac{2\pi ml}{M}}$. In MRC, $\bD_m$ is a $K \times K$ diagonal matrix with the $k^{\rm th}$ diagonal element given by $D_m^{k,k} = \sum_{i=0}^{N-1} |H_m^{i,k}|^2$. The role of $\bD_m$ is to normalize the amplitude of the MRC output. Without this term, the amplitude grows linearly without a bound as the number of BS antennas increases.

We note that for large number of BS antennas $N$ and using the law of large numbers, $\bD_m$ tends to $N \eye_K$. Similarly, when $N$ grows large and due to the law of large numbers, $\bH_m^{\rm H} \bH_m$ tends to $N \eye_K$, \cite{ngo2013energy}.  Hence, all of the above combiners tend to $\frac{1}{N} \bH_m$, i.e., matched filter, as the number of BS antennas increases, \cite{ngo2013energy}. Therefore, in the following, to find the various interference terms in the \emph{asymptotic} regime, i.e., as the number of BS antennas $N$ approaches infinity, we consider matched filter (MF) multi-antenna combining according to $\BW_m = \frac{1}{N} \bH_m$.

Before we continue, we recall the following result from probability theory, paving the way for our upcoming derivations. Let $\ba = [a_1,\dots,a_n]^{\rm T}$ and $\bb = [b_1,\dots,b_n]^{\rm T}$ be two random vectors each containing i.i.d. elements. Moreover, assume that the $i^{\rm th}$ elements of $\ba$ and $\bb$ are correlated according to $\mathbb{E}\big\{ a_i^* b_i \big\} = C_{ab}$, $i = 1,\dots,n$. Consequently, according to the law of large numbers, the sample mean $\frac{1}{n} \ba^{\rm H} \bb = \frac{1}{n} \sum_{i=1}^n a_i^* b_i$ converges almost surely to the distribution mean $C_{ab}$ as $n$ tends to infinity.

In the asymptotic regime, i.e., as $N$ tends to infinity, the elements of $\bG_{mm',nn'} = \BW_m^{\rm H} \bH_{mm',nn'}$ can be calculated using the law of large numbers. In particular, as $N$ grows large, the element $(k,k')$ of $\bG_{mm',nn'}$ converges almost surely to
\be \label{eqn:G_asymp}
G_{mm',nn'}^{k,k'} \rightarrow \mathbb{E} \Big\{ \left( H_{m}^{i,k} \right)^* H_{mm',nn'}^{i, k'} \Big\} .
\ee
To calculate the right hand side of (\ref{eqn:G_asymp}), we first find the equivalent time-domain channel impulse response \emph{after} multi-antenna combining. In particular, let $g_{mm'}^{k,k'}[n]$ denote the equivalent channel impulse response between the transmitted symbols at subcarrier $m'$  of terminal $k'$ and the received ones at subcarrier $m$ of BS output corresponding to terminal $k$ after combining\footnote{Note that we have used the letters $g$ and $G$, respectively, to denote the equivalent time and frequency channel coefficients \emph{after combining}. On the other hand, letters $h$ and $H$ have been used in (\ref{eqn:mimo_equiv_chan}), to refer to the respective channel coefficients \emph{before combining}.}. Following (\ref{eqn:mimo_equiv_chan}), we have
\begin{align} \label{eqn:equiv_chan_g} 
g_{mm'}^{k,k'}[n] = \frac{1}{N} \sum_{i=0}^{N-1} \big( H_{m}^{i,k} \big)^* \big(f_{m'}[l] \star h_{i, k'}[l] \star  f_m^\ast[-l]\big)_{\downarrow \frac{M}{2}} . 
\end{align} 
Hence, as the number of BS antennas grows large,  the asymptotic equivalent channel response can be obtained using the law of large numbers according to
\begin{align} \label{eqn:asym_equiv_chan} 
g_{mm'}^{k,k'}[n] &\rightarrow \mathbb{E} \Big\{ \left( H_{m}^{i,k} \right)^* \big(f_{m'}[l] \star h_{i, k'}[l] \star f_m^\ast[-l]\big)_{\downarrow \frac{M}{2}} \Big\} \nonumber \\
&=  \Big( f_{m'}[l] \star \mathbb{E} \Big\{ \left( H_{m}^{i,k} \right)^* h_{i,k'}[l] \Big\} \star f_m^\ast[-l]\Big)_{\downarrow \frac{M}{2}} .
\end{align} 
The above expression includes a correlation between the channel frequency coefficient $H_m^{i,k}$ and the channel impulse response $h_{i,k'}[l]$. This correlation can be calculated as
\begin{align} \label{eqn:channel_corr}
\mathbb{E} \Big\{ \left( H_{m}^{i,k} \right)^* h_{i, k'}[l] \Big\} &= \sum_{\ell = 0}^{L_{\rm h}-1} \mathbb{E} \left\{ h_{i,k}^*[\ell] h_{i,k'}[l] \right\} e^{j\frac{2 \pi \ell m}{M}}  \nonumber \\
&= p_k[l] e^{j\frac{2 \pi l m}{M}} \delta_{kk'} = p_{k,m}[l] \delta_{kk'} ,
\end{align}
where $p_{k,m}[l] \triangleq p_k[l] e^{j\frac{2 \pi l m}{M}}$ is the channel PDP of terminal $k$ modulated to the center frequency of the $m^{\rm th}$ subcarrier. The result in (\ref{eqn:channel_corr}) shows the correlation between the combiner taps at the receiver and the channel impulse responses between MTs and the BS antennas. The following proposition states the impact of this correlation on the SINR at the receiver outputs.

\begin{proposition} \label{prp:saturation1}
In an FBMC massive MIMO system, as the number of BS antennas tends to infinity, the effects of multiuser interference and noise vanish. However, some residual ISI and ICI from the same user remain even with infinite number of BS antennas. In particular, for a given user $k$, the equivalent channel impulse response between the transmitted data symbols at subcarrier $m'$ and the received ones at subcarrier $m$ tends to
\be \label{eqn:MRC_equiv_response}
g_{mm'}^{k,k}[n] \rightarrow \Big( f_{m'}[l] \star p_{k,m}[l] \star f_m^\ast[-l]\Big)_{\downarrow \frac{M}{2}} ,
\ee
which is dependent on the channel PDP. As a result, the SINR converges almost surely to 
\be \label{eqn:SINR_sat}
{\rm SINR}^k_{m,n} \rightarrow \frac{ \Re^2 \big\{ G^{k,k}_{mm,nn} \big\} }{ \mathop{\sum\limits_{n'=-\infty}^{+\infty} \sum\limits_{m'=0}^{M-1}}\limits_{(m',n')\neq(m,n)} \Re^2 \big\{ {G}^{k,k}_{mm',nn'} \big\}} ,
\ee
where $G^{k,k}_{mm',nn'} = {g}^{k,k}_{mm'}[n-n'] \hspace{2pt} e^{j(\theta_{m',n'} - \theta_{m,n})}$. The above value constitutes an upperbound for the SINR performance of the system. Hence, arbitrarily large SINR values cannot be achieved by increasing the BS array size.
\end{proposition}
\begin{proof}
As suggested by (\ref{eqn:channel_corr}), when $k'\neq k$, the channel response tends to zero. Thus, multiuser interference tends to zero. A similar argument can be made for the additive noise. This results from the law of large numbers and the fact that the combiner coefficients are uncorrelated with the filtered noise samples. When $k' = k$, which implies the interference from the same user on itself, the channel response tends to (\ref{eqn:MRC_equiv_response}). Notice that due to the presence of $p_{k,m}[l]$, the orthogonality condition of (\ref{eqn:orthogonality}) does not hold anymore even with an infinite number of BS antennas. Consequently, some residual ISI and ICI remain and cause the SINR to saturate at a deterministic level given in (\ref{eqn:SINR_sat}). 
\end{proof}

We note that according to (\ref{eqn:G_asymp}), the asymptotic SINR saturation results from the statistical correlation between the multi-antenna combiner taps and the interference coefficients. This correlation is an inherent property of FBMC-based massive MIMO systems and is due to the transients of the channel impulse response since no cyclic prefix (CP) is used. In particular, when the multi-antenna combining is performed in the frequency domain according to (\ref{eqn:mrc_zf_mmse}), such correlation appears as a result of the leakage due to the absence of CP. This result is general as a similar phenomenon also emerges in massive MIMO systems based on OFDM without CP, \cite{aminjavaheri2017ofdm}.

\section{Equalization} \label{sec:prototype_modify}

As discussed in the previous section, even with an infinite number of BS antennas, some residual ICI and ISI remain due to the correlation between the combiner taps and the channel impulse responses between the MTs and the BS antennas. As a solution to this problem, in this section, we propose an efficient equalization method to remove the above correlation.

In (\ref{eqn:MRC_equiv_response}), the problematic term that leads to the saturation issue is the modulated channel PDP, $p_{k,m}[l]$. In the absence of this term, the channel response $g^{k,k}_{mm'}[n] = \big( f_{m'}[l] \star f^*_m[-l] \big)_{\downarrow \frac{M}{2}}$ does not incur any interference provided that $q[l] = f[l] \star f^*[-l]$ is a Nyquist pulse. This observation suggests that we can resolve the saturation issue by equalizing the effect of $p_{k,m}[l]$. Let $P_k(\omega)$ denote the discrete-time Fourier transform (DTFT) of $p_k[l]$. Similarly, we define $P_{k,m}(\omega) = P_k(\omega-2\pi m/M)$ as the DTFT of $p_{k,m}[l]$. This observation implies that one can equalize the effect of $p_{k,m}[l]$ by introducing a filter $\phi_{k,m}[l]$ with transfer function
\be \label{eqn:Phi}
\Phi_{k,m}(\omega) = \frac{1}{P_{k,m}(\omega)} ,
\ee
in cascade with $f_m^*[-l]$ to achieve the desired equivalent channel response $g^{k,k}_{mm'}[n] \rightarrow \big( f_{m'}[l] \star f^*_m[-l] \big)_{\downarrow \frac{M}{2}}$ in the asymptotic regime. This modifies the receiver structure as illustrated in Fig.~\ref{fig:eq_receiver1}. 

\begin{figure*}[!t]
\centering
\includegraphics[scale=0.75]{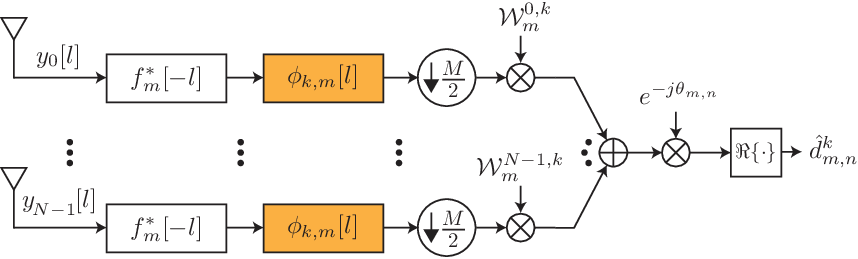}
\caption{Block diagram of the proposed receiver structure to resolve the saturation issue. Here, only the portion of the receiver corresponding to subcarrier $m$ and terminal $k$ is shown.}
\label{fig:eq_receiver1}
\end{figure*}

\begin{figure*}[!t]
\centering
\includegraphics[scale=0.75]{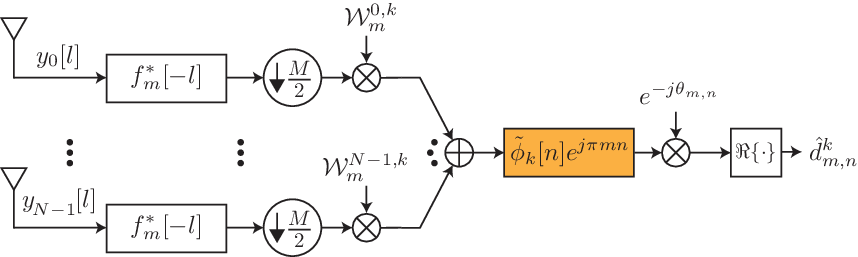}
\caption{Block diagram of the simplified receiver. Utilizing multi-rate signal processing techniques, the additional equalization block can be moved to after the analysis filter bank and combiner to minimize the computational cost.}
\label{fig:eq_receiver2}
\end{figure*}

\begin{proposition} \label{prp:self_eq}
In an FBMC massive MIMO system, as the number of BS antennas tends to infinity and by using the proposed equalization method, the channel distortions, i.e., ICI and ISI, as well as MUI and noise effects will disappear, and arbitrarily large SINR performances can be achieved.
\end{proposition}
\begin{proof} 
Using the equalizer in (\ref{eqn:Phi}), the distortion due to the channel PDP $p_{k,m}[l]$ in the equivalent channel impulse response in (\ref{eqn:MRC_equiv_response}) is removed. Hence, the equivalent channel impulse response tends to that of an ideal channel. As a result, the effects of ICI and ISI will vanish asymptotically. 

Note that in the presence of the proposed equalizer, multiuser interference still tends to zero. This is due to the fact that the asymptotic values of the multiuser interference coefficients are given by (\ref{eqn:G_asymp}) for $k\neq k'$. Since the channels of different users are independent, the effect of multiuser interference tends to zero whether or not the proposed equalizer is in place. This argument also holds for the noise contribution since the combining coefficients and the filtered noise samples are independent.
\end{proof}

It is worth mentioning that in the above analysis, we did not make any assumption about the flatness of the channel response over the bandwidth of each subcarrier. Thus, the result obtained in Proposition \ref{prp:self_eq} is valid for any frequency-selective channel. It is worth mentioning that according to (\ref{eqn:Phi}), the proposed filter response depends on the channel PDPs. Hence, the BS needs to estimate the channel PDP for each terminal to be able to avoid the saturation issue. Fortunately, in massive MIMO systems, the channel PDP can be estimated in a relatively easy and feasible manner. In particular, the channel PDP for each terminal can be determined by calculating the mean power of each tap of the respective channel impulse responses across different BS antennas. As the number of BS antennas increases, according to the law of large numbers, this estimate becomes closer to the exact channel PDP.

Although the above method resolves the saturation problem, it may not be of practical interest as it may lead to a very complex receiver. The source of the complexity lies in the requirement of a separate filter $\phi_{k,m}[l]$ per user per antenna. Hence, the receiver front-end processing has to be repeated for each terminal separately. Next, we utilize multi-rate signal processing techniques and propose the following steps to resolve the complexity issue.

\begin{proposition}\label{prp:filter_swap} 
In an FBMC massive MIMO system, the channel PDP equalization can be performed after analysis filter bank and combiner as in Fig. \ref{fig:eq_receiver2}. Here, 
\be \label{eqn:Phi_tilde}
\tilde{\phi}_k[n] \triangleq \Big( \phi_k[l] \star \sinc(2l/M) \Big)_{\downarrow \frac{M}{2}} ,
\ee
where $\phi_k[l] \triangleq \phi_{k,0}[l]$ and $\sinc(t) \triangleq \frac{\sin(\pi t)}{\pi t}$. Note that the term $\sinc(2l/M)$ acts as an ideal low-pass filter with bandwidth $\frac{2\pi}{M}$.  
\end{proposition}
\begin{proof}
The FBMC prototype filter is normally designed such that its frequency response is almost perfectly confined to the interval $[-\frac{2\pi}{M},\frac{2\pi}{M}]$. Hence, in Fig.~\ref{fig:eq_receiver1}, after filtering the incoming signal $y_m[l]$ by $f_m^*[-l]$, the frequency response of the result is almost perfectly confined to the frequency interval $\big[\frac{2\pi(m-1)}{M},\frac{2\pi(m+1)}{M}\big]$. This implies that the input to $\phi_{k,m}[l]$ is band-limited. It is intuitive that since the input to the equalizer is band-limited, the equalization processing can take place in the low rate (after decimation). Subsequently, the filtering can be moved to after the combining due to the linearity. This leads to the structure in Fig.~\ref{fig:eq_receiver2}. Note that the equalizer used for any particular subcarrier can be obtained from the one used for subcarrier $0$. In the following, we rigorously prove that the equalization can be performed after the decimation.

\begin{figure}[!t]
\centering
\includegraphics[scale=0.75]{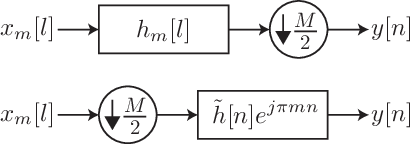}
\caption{Two equivalent systems considered in the proof of Proposition~\ref{prp:filter_swap}.}
\label{fig:filter_swap}
\end{figure}

For simplicity, consider the two systems given in Fig.~\ref{fig:filter_swap}. Here, $x[l]$ is an arbitrary band-limited signal whose spectrum is confined to the frequency interval $[-\frac{2\pi}{M},\frac{2\pi}{M}]$, and $h[l]$ is an arbitrary filter impulse response. Moreover, let $x_m[l] \triangleq x[l]e^{j\frac{2\pi ml}{M}}$ and $h_m[l] \triangleq h[l]e^{j\frac{2\pi ml}{M}}$ represent the modulated versions of $x[l]$ and $h[l]$, respectively, and $\tilde{h}[n] \triangleq \big( h[l] \star \sinc(2l/M) \big)_{\downarrow \frac{M}{2}}$ denote the band-limited and decimated version of $h[l]$. We prove that two systems shown in Fig.~\ref{fig:filter_swap} are equivalent.

First consider the top system in Fig.~\ref{fig:filter_swap}, and let $\hat{h}[l] \triangleq h[l] \star \frac{2}{M}\sinc(2l/M)$ which has the transfer function
\be
\hat{H}(\omega) = 
 \left\{
  \begin{array}{lr}
    H(\omega), & \omega \in [-\frac{2\pi}{M},\frac{2\pi}{M}], \\
    0,         & \text{else} .
  \end{array} 
\right. . \nonumber
\ee
Note that since the input signal does not have any frequency component outside of the frequency interval $\big[\frac{2\pi(m-1)}{M},\frac{2\pi(m+1)}{M}\big]$, it is possible to use the filter $\hat{h}_m[l]\triangleq \hat{h}[l]e^{j\frac{2\pi ml}{M}}$ instead of $h_m[l]$ in the top system in Fig.~\ref{fig:filter_swap}. Subsequently, after the decimation operation, the DTFT of the output signal $y[n]$ can be expressed as, \cite{vetterli2014foundations}, 
\be
Y(\omega) \hspace{-2pt} = \hspace{-3pt} \frac{2}{M} \hspace{-5pt} \sum_{k=0}^{\frac{M}{2}-1} \hspace{-5pt} X \hspace{-1pt}\Big(\frac{2\omega-2\pi(2k+m)}{M}\Big) \hat{H} \hspace{-1pt} \Big(\frac{2\omega-2\pi(2k+m)}{M}\Big) . \nonumber
\ee
Using the fact that both $X(\omega)$ and $\hat{H}(\omega)$ are band-limited to $[-\frac{2\pi}{M},\frac{2\pi}{M}]$, we find that in the summation above, only one of the terms is non-zero. In particular, for \textit{even} $m$ we have
\be
Y(\omega) = \frac{2}{M} X\Big(\frac{2\omega}{M}\Big) \hat{H}\Big(\frac{2\omega}{M}\Big), ~~~~ -\pi \leq \omega \leq +\pi , \nonumber
\ee
and for \textit{odd} $m$ we have
\be
Y(\omega) = \frac{2}{M} X\Big(\frac{2\omega-2\pi}{M}\Big) \hat{H}\Big(\frac{2\omega-2\pi}{M}\Big), ~~~~ 0 \leq \omega \leq 2\pi . \nonumber
\ee
Here, it is worth to mention that when $m$ is even, $\frac{2}{M} X\big(\frac{2\omega}{M}\big)$ and $\frac{2}{M} \hat{H}\big(\frac{2\omega}{M}\big)$ represent the DTFT of the decimated versions of $x_m[l]$ and $\hat{h}_m[l]$, respectively. Similarly, when $m$ is odd, $\frac{2}{M} X\big(\frac{2\omega-2\pi}{M}\big)$ and $\frac{2}{M} \hat{H}\big(\frac{2\omega-2\pi}{M}\big)$ express the DTFT of the decimated versions of $x_m[l]$ and $\hat{h}_m[l]$, respectively. Consequently, instead of passing $x_m[l]$ through the filter $\hat{h}_m[l]$ and decimating the result, one can decimate both $x_m[l]$ and $\hat{h}_m[l]$ separately, and then convolve them together in the low rate. Before we finish the proof, we just aim to derive the decimated version of $\hat{h}_m[l]$ in terms of $h[l]$. We have
\begin{align*}
\frac{M}{2} \left( \hat{h}_m[l] \right)_{\downarrow \frac{M}{2}} &= \frac{M}{2} \left( \left( h[l] \star \frac{2}{M} \sinc(2l/M) \right) e^{j\frac{2\pi ml}{M}}  \right)_{\downarrow \frac{M}{2}} \\
&= \tilde{h}[n] e^{j\pi mn}  .
\end{align*}
This results in the system given in Fig.~\ref{fig:filter_swap}. This completes the proof.
\end{proof}

As suggested by the above proposition, one can incorporate the receiver structure shown in Fig.~\ref{fig:eq_receiver2} to resolve the saturation issue in an efficient manner. In particular, after the analysis filter bank and multi-antenna combining, the filter $\tilde{\phi}_k[n] e^{j\pi m n}$ can be incorporated to equalize the effect of the problematic term $p_{k,m}[l]$ in (\ref{eqn:MRC_equiv_response}). Note that in this approach, the main parts of the receiver front-end including the analysis filter bank and the multi-antenna combiner will remain unchanged. The advantages of this simplified structure as compared to the previous one include: (i) The analysis filter bank is common for all terminals and can be performed once. (ii) The additional equalizer has a very short length since it is performed at the low rate after decimation, and (iii) the equalizer is performed after the multi-antenna combining, hence, its computational cost is independent of the number of BS antennas.

Before we end this section, we note that according to (\ref{eqn:MRC_equiv_response}), a frequency shifted version of the power delay profile $p_{k}[l]$ distorts the equivalent channel. As a result, only the frequency response of $p_k[l]$ limited to the interval $\omega \in \left[ -\frac{2\pi}{M},+\frac{2\pi}{M}\right]$ affects the respective equivalent channel response. This interval corresponds to the width of a single subcarrier.

\section{Frequency-Domain Perspective} \label{sec:self-equalization}


In this section, we aim at studying the results of the previous sections from the frequency-domain point of view. As we show, this study leads to a deeper understanding of FBMC in massive MIMO channels.

In OFDM-based systems, presence of the CP greatly simplifies the equalization procedure. In particular, as long as the length of the CP is larger than the duration of channel impulse response, one can utilize a single-tap equalizer per subcarrier to undo the effect of the channel and retrieve the transmitted data symbols. On the other hand, in FBMC-based systems, since no CP is adopted, single-tap equalization does not fully compensate the  channel frequency-selectivity across subcarrier bands. However, assuming that the number of subcarriers is sufficiently large so that the channel frequency response is approximately flat over each subcarrier band, then the model described by (\ref{eqn:siso_demod_symbol2}) is going to be valid. Therefore, the task of equalization can be simplified by using single-tap equalization per subcarrier.

In this section, we aim at discussing the fact that in massive MIMO systems, by using the equalization method developed in Section \ref{sec:prototype_modify}, it is not necessary to have a flat channel response over the band of each subcarrier in order to use single-tap equalizer. In particular, by using the simple single-tap per subcarrier equalization even in strong frequency selective channels and by incorporating a large number of antennas at the BS, the effective channel response becomes flat. It is clear that this property has a number of advantages from the system implementation point of view. In particular, since there is no need for flat-fading assumption over the band of each subcarrier, one can \emph{widen} the subcarrier widths (or equivalently decrease the symbol duration). Consequently, the following advantages can be achieved, \cite{armanfarhang2014filter}. 
\begin{enumerate}
\item The sensitivity to carrier frequency offset (CFO) in the uplink of multiple access networks is decreased by widening the subcarrier bands.
\item The peak-to-average power ratio (PAPR) is lowered, which leads to larger coverage and higher battery efficiency in mobile terminals. This is a direct consequence of reducing the number of subcarriers in a synthesized signal.
\item The sensitivity to channel time variations within the FBMC symbol duration is reduced. This advantage arises from the reduction of the symbol duration. As a result, a higher quality of service is expected in highly time-varying channels such as in high speed trains. 
\item The latency between the terminals and the BS is decreased, as a result of shorter symbol durations. This is crucial for addressing the low-latency requirements of the 5G networks.
\item The inefficiency due to the ramp-up and ramp-down of the prototype filter at the beginning and the end of each packet, especially in bursty communications, is decreased. This results from the shortening of the symbol duration which in turn leads to a shorter prototype filter in the time domain, \cite{farhang2016comparison}.
\end{enumerate}

\begin{figure*}[!t]
\centering
\subfigure[]{\includegraphics[scale=0.68]{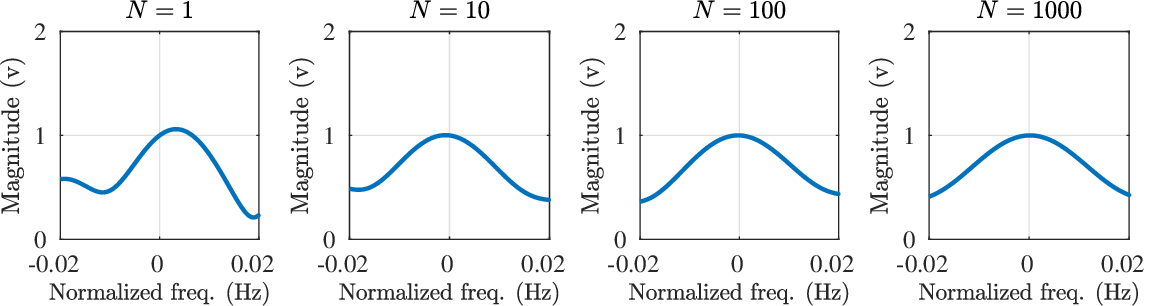}%
\label{fig:self_eqa}}
\par\medskip
\subfigure[]{\includegraphics[scale=0.68]{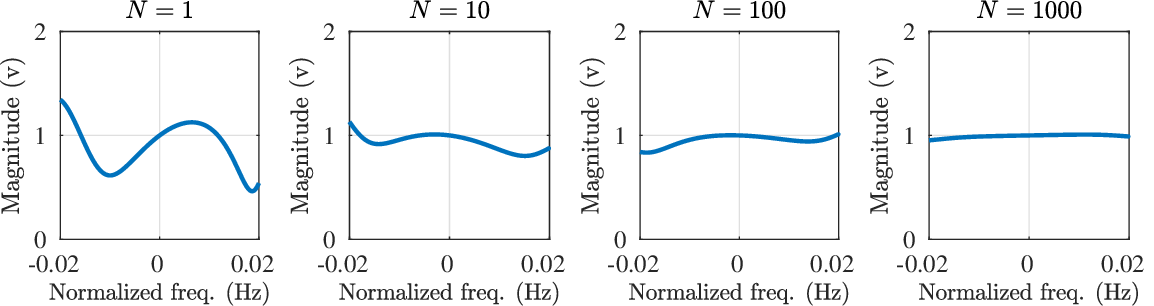}%
\label{fig:self_eqb}}
\caption{Illustration of the equivalent channel response. Here, we assume $M=512$, and consider an exponentially decaying channel PDP with the decaying factor of $0.06$ and the length of $L_{\rm h} = 50$. (a) The equivalent channel, $C_m^{k,k}(\omega)$, for subcarrier $m=0$, without the proposed equalizer. (b) The equivalent channel, $\tilde{C}_m^{k,k}(\omega)$, for subcarrier $m=0$, with the proposed equalizer. As the number of BS antennas increases, the equivalent channel becomes flat only when the proposed equalizer is in place.}
\label{fig:self_eq}
\end{figure*}


Following (\ref{eqn:equiv_chan_g}), we can obtain the frequency response of the equivalent channel after combining. To this end, consider a given terminal $k$, and let ${G}_{mm'}^{k,k}(\omega)$ denote the frequency response of the high-rate (i.e., without decimation) equivalent channel between the transmitted symbols at subcarrier $m'$ and the received ones at subcarrier $m$. We have
\begin{align} \label{eqn:Homega}
{G}_{mm'}^{k,k}(\omega) &= \frac{1}{N}  \sum_{i=0}^{N-1} \left( H_{m}^{i,k} \right)^* F_{m'}(\omega) \hspace{1pt} H_{i, k}(\omega) \hspace{1pt} F_m^\ast(\omega)  \nonumber \\
&= C_m^{k,k}(\omega) F_{m'}(\omega) F_m^\ast(\omega) ,
\end{align}
where 
\be \label{eqn:C}
C_m^{k,k}(\omega) = \frac{1}{N}  \sum\limits_{i=0}^{N-1} \left( H_{m}^{i,k} \right)^* H_{i,k}(\omega) .
\ee
In (\ref{eqn:Homega}), $F_{m'}(\omega)$ and $F_m^*(\omega)$ are two modulated square-root Nyquist filters, i.e., $Q(\omega) = |F(\omega)|^2$ is a Nyquist pulse, and ${C}_m^{k,k}(\omega)$ is due to the \emph{multipath channel}. Ideally, $C_m^{k,k}(\omega)$ should be flat over the pass band of the subcarrier $m$ so that the symbols of subcarrier $m$ can be perfectly reconstructed without any interference. However, when there exists a frequency-selective channel, the term $C_m^{k,k}(\omega)$ may incur some distortion over the pass band of subcarrier $m$ and, accordingly, lead to some interference in the detected symbols. As the number of BS antennas grows large, using the law of large numbers and according to (\ref{eqn:channel_corr}), $C_m^{k,k}(\omega)$ tends to $P_{k,m}(\omega)$. Therefore, the flat-fading condition may not be achieved by just increasing the BS array size.

On the other hand, when the equalizer in (\ref{eqn:Phi}) is utilized, the equivalent channel in the frequency domain can be expressed as 
\be
\tilde{G}_{mm'}^{k,k}(\omega) = \tilde{C}_m^{k,k}(\omega) F_{m'}(\omega) F_m^\ast(\omega),
\ee
where
\begin{align} \label{eqn:Ctilde}
\tilde{C}_m^{k,k}(\omega) = \frac{ \frac{1}{N} \sum\limits_{i=0}^{N-1} \left( H_{m}^{i,k} \right)^* H_{i,k}(\omega)}{ P_{k,m}(\omega)} = \frac{C_m^{k,k}(\omega)}{P_{k,m}(\omega)}.
\end{align}
Therefore, since $C_m^{k,k}(\omega)$ asymptotically tends to $P_{k,m}(\omega)$, $\tilde{C}_m^{k,k}(\omega)$ will in turn tend to a frequency flat channel. Thus, no interference is expected in large antenna regime. This channel flattening effect of FBMC-based massive MIMO systems is illustrated in Fig.~\ref{fig:self_eq}.

\section{SINR Analysis} \label{sec:sinr}

In this section, we analyze the SINR performance of an FBMC-based massive MIMO system in the uplink incorporating the proposed equalization method. We limit our study to the two most prominent linear combiners namely, MRC and ZF. As mentioned earlier, in the large antenna regime, all the combiners in (\ref{eqn:mrc_zf_mmse}) tend to $\frac{1}{N} \bH_m$, and hence, the same asymptotic SINR performance as in MRC and ZF is expected for the MMSE combiner. As mentioned earlier, the equalization approaches given in Figs.~\ref{fig:eq_receiver1} and \ref{fig:eq_receiver2} are equivalent. Although the method given in Fig.~\ref{fig:eq_receiver2} is preferred for implementation, here, for the purpose of analysis, we consider the approach given in Fig.~\ref{fig:eq_receiver1}. Throughout this section, we consider normalized channel PDPs for each terminal such that 
\be \label{eqn:pdp_normalization}
\sum_{l=0}^{L_{\rm h}-1} p_k[l] = 1, ~~~~ k \in \{0,\dots,K-1\} .
\ee

In Fig.~\ref{fig:eq_receiver1}, the receiver filter $f^*_m[-l]$ and the equalizer $\phi_{k,m}[l]$ can be combined together as a single filtering block with impulse response $\tilde{f}^*_{k,m}[-l] \triangleq f^*_m[-l] \star \phi_{k,m}[l]$. Therefore, we can consider having the new receiver filter $\tilde{f}^*_{k,m}[-l]$ in place, and use (\ref{eqn:mimo_est_d1}) to obtain the estimated data symbols. To this end, let $\tilde{\bH}_{mm',nn'}^k$ be an $N \times K$ matrix with elements given by (\ref{eqn:mimo_equiv_chan}) but with the new filter $\tilde{f}^*_{k,m}[-l]$ in place instead of $f^*_{m}[-l]$. Moreover, we form the $K \times K$ matrix $\tilde{\bG}_{mm',nn'}$ similar to $\bG_{mm',nn'}$. In particular, the $k^{\rm th}$ row of $\tilde{\bG}_{mm',nn'}$ can be calculated as $\bw_{m,k}^{\rm H} \tilde{\bH}_{mm',nn'}^k$, where $\bw_{m,k}$ is the $k^{\rm th}$ column of the combiner matrix $\BW_m$. Following the above definitions, the interference coefficients are determined by the real part of the elements of $\tilde{\bG}_{mm',nn'}$. In order to pave the way for our SINR analysis, we desire to find the elements of $\tilde{\bG}_{mm',nn'}$ in a matrix form. Towards this end and based on (\ref{eqn:mimo_equiv_chan}) and (\ref{eqn:mimo_est_d1}), the convolution, downsampling, multi-antenna combining, and phase compensation operations can all be expressed compactly as
\be \label{eqn:G_elements}
\tilde{G}_{mm',nn'}^{k,k'} = \big( \bpsi_{mm',nn'}^k \big)^{\rm H} \hspace{3pt} \bg_m^{k,k'} , 
\ee
where
\be \label{eqn:g}
\bg^{k,k'}_m = \sum_{i=0}^{N-1} \left( \mathcal{W}_{m}^{i,k} \right)^* \bh_{i,k'} ,
\ee
and
\be 
\big( \bpsi_{mm',nn'}^k \big)^{\rm H} = e^{j(\theta_{m',n'} - \theta_{m,n})} \hspace{2pt} \e_{nn'}^{\rm T} \tilde{\bF}_{k,m}  \bF_{m'} .
\ee
The vector $\bg^{k,k'}_m$ is the effective multipath channel impulse response between terminals $k$ and $k'$ at subcarrier $m$, after the combining operation. $\bh_{i,k} \triangleq \big[h_{i,k}[0], \dots, h_{i,k}[L_{\rm h}-1]\big]^{\rm T}$ is the vector of channel impulse response between $i^{\rm th}$ BS antenna and $k^{\rm th}$ terminal. $\bF_{m'}$ and $\tilde{\bF}_{k,m}$ are two Toeplitz matrices that are defined in (\ref{eqn:Pm}) and (\ref{eqn:Ptildem}), respectively, and signify the synthesis filter at subcarrier $m'$ and the new analysis filter at subcarrier $m$, respectively. Note that the size of the matrix $\bF_{m'}$ is $(L_{\rm f} + L_{\rm h} - 1) \times L_{\rm h}$. To determine the size of $\tilde{\bF}_{k,m}$, we follow (\ref{eqn:Phi}) to note that $f_m[l] = \tilde{f}_{k,m}[l] \star p_{k,m}^*[-l]$. Hence, the length of the new filter $\tilde{f}_{k,m}[l]$ can be obtained as $L_{\tilde{\rm f}} = L_{\rm f} - L_{\rm h} + 1$. As a result, the size of $\tilde{\bF}_{k,m}$ can be calculated as $(2L_{\rm f}-1) \times (L_{\rm f}+L_{\rm h}-1)$. The $(2L_{\rm f}-1) \times 1$ vector $\e_{nn'}$ is accounted for the downsampling operation and contains zeros except on its $(L_{\rm f}+(n-n')\frac{M}{2})^{\rm th}$ entry which is equal to one. Finally, $e^{j(\theta_{m',n'} - \theta_{m,n})}$ is due to the phase compensation. 
\bse
\begin{align} \setlength{\arraycolsep}{4pt}
\mathbf{F}_{m'} \hspace{-0.5mm} = \hspace{-1mm} \begin{pmatrix}
f_{m'}[0] & 0 & \cdots & 0 & 0 \\ 
f_{m'}[1] & f_{m'}[0] & \cdots & 0 & 0 \\
\vdots & \vdots & \vdots & \vdots & \vdots\\
0 & 0 & \cdots &  f_{m'}[L_{\rm f}-1] & f_{m'}[L_{\rm f}-2]\\ 
0 & 0 & \cdots &  0 & f_{m'}[L_{\rm f}-1]
\end{pmatrix} , \label{eqn:Pm} 
\end{align}
\begin{align} \setlength{\arraycolsep}{2pt}
\tilde{\mathbf{F}}_{k,m} = \begin{pmatrix}
\tilde{f}^*_{k,m}[L_{\tilde{\rm f}}-1] & 0 & \cdots & 0 & 0 \\ 
\tilde{f}^*_{k,m}[L_{\tilde{\rm f}}-2] & \tilde{f}^*_{k,m}[L_{\tilde{\rm f}}-1] & \cdots & 0 & 0 \\
\vdots & \vdots & \vdots & \vdots & \vdots\\
0 & 0 & \cdots &  \tilde{f}^*_{k,m}[0] & \tilde{f}^*_{k,m}[1]\\ 
0 & 0 & \cdots &  0 & \tilde{f}^*_{k,m}[0]
\end{pmatrix} . \label{eqn:Ptildem}
\end{align}
\ese

Note that in (\ref{eqn:G_elements}), the term $ \bpsi_{mm',nn'}^k $ is completely deterministic, whereas $\bg_{m}^{k,k'}$ is a random vector. Therefore, in this equation, we have decomposed the interference coefficients into random and deterministic components. Moreover, while $\bpsi_{mm',nn'}^k $ does not depend on the type of combining, $\bg_{m}^{k,k'}$ is directly related to the combining method and should be evaluated for each combiner separately.

\subsection{MRC} \label{sec:sinr_mrc}

In MRC, as the number of BS antennas grows large, $\bD_m$ in (\ref{eqn:mrc_zf_mmse}) tends to $N \eye_K$. Therefore, we can write $\bg_m^{k,k'} =  \frac{1}{N} \sum_{i=0}^{N-1} \left( H_{m}^{i,k} \right)^* \bh_{i,k'}$. In the Appendix, we have calculated the first and second order statistics of the complex random vector $\bg_m^{k,k'}$. The result is
\bse \label{eqn:g_mrc_stats} 
\begin{flalign} 
\bmu_m^{k,k'} \hspace{-4pt} &\triangleq \hspace{-2pt} \mathbb{E} \big\{ \bg^{k,k'}_m \big\} = \delta_{kk'} \bp_{k,m}, \label{eqn:g_mrc_mean}\\
\bGamma_{m}^{k,k'} \hspace{-4pt} &\triangleq \hspace{-2pt} \mathbb{E} \Big\{ \hspace{-2pt} \big( \bg^{k,k'}_m  \hspace{-3pt} -  \hspace{-2pt} \bmu_m^{k,k'} \big) \hspace{-1pt} \big(\bg^{k,k'}_m \hspace{-3pt} - \hspace{-2pt}  \bmu_m^{k,k'} \big)^{\mathrm{H}}\Big\} \hspace{-2pt} = \hspace{-2pt} \frac{1}{N} \bD_{{\rm p}_{k'}}, \label{eqn:g_mrc_cov} \\
\bK_{m}^{k,k'} \hspace{-4pt} &\triangleq \hspace{-2pt} \mathbb{E} \Big\{ \hspace{-2pt} \big( \bg^{k,k'}_m \hspace{-3pt} - \hspace{-2pt} \bmu_m^{k,k'} \big) \hspace{-1pt} \big(\bg^{k,k'}_m \hspace{-3pt} - \hspace{-2pt} \bmu_m^{k,k'} \big)^{\mathrm{T}}\Big\} \hspace{-2pt} = \hspace{-2pt} \frac{1}{N} \delta_{kk'} \bp_{k,m} \bp_{k,m}^{\rm T}, \label{eqn:g_mrc_rel}
\end{flalign} 
\ese 
where $\bD_{{\rm p}_k} \triangleq {\rm diag} \big\{\big[p_k[0],p_k[1],\dots,p_k[L_{\rm h}-1]\big]^{\rm T}\big\}$, and $\bp_{k,m} \triangleq \big[p_{k,m}[0],p_{k,m}[1],\dots,p_{k,m}[L_{\rm h}-1]\big]^{\rm T}$.

Let $\bgamma^{k,k'}_m$ be a zero-mean random vector defined as ${\bgamma}^{k,k'}_m \triangleq \bg^{k,k'}_m - \bmu^{k,k'}_m$ . Thus, from (\ref{eqn:G_elements}) and (\ref{eqn:g_mrc_mean}) we have 
\begin{align} \label{eqn:G_elements2}
&\tilde{G}_{mm',nn'}^{k,k'}  \nonumber \\
&= \big( \bpsi_{mm',nn'}^k \big)^{\rm H} \bgamma^{k,k'}_m + \delta_{kk'} \big( \bpsi_{mm',nn'}^k \big)^{\rm H} \bp_{k,m} \nonumber \\
&= \big( \bpsi_{mm',nn'}^k \big)^{\rm H} \bgamma^{k,k'}_m + \delta_{kk'} \big( \delta_{mm'} \delta_{nn'} + jA_{mm',nn'} \big) ,
\end{align}
where $A_{mm',nn'} \triangleq \Im \big\{ \sum_{l=-\infty}^{+\infty} a_{m',n'}[l] a^*_{m,n}[l] \big\}$. The second line of (\ref{eqn:G_elements2}) follows from the real-orthogonality property of FBMC given in (\ref{eqn:orthogonality}). We recall that by incorporating the equalizer $\phi_{k,m}[l]$, the effect of the modulated channel PDP $p_{k,m}[l]$ is removed and the real-orthogonality condition is satisfied. Hence, the term $ \big( \bpsi_{mm',nn'}^k \big)^{\rm H} \bp_{k,m} = e^{j(\theta_{m',n'} - \theta_{m,n})} \hspace{2pt} \e_{nn'}^{\rm T} \tilde{\bF}_{k,m}  \bF_{m'} \bp_{k,m}$ is equal to $\delta_{mm'} \delta_{nn'} + j A_{mm',nn'}$ since the matrix $\tilde{\bF}_{k,m}$ compensates the effect of $\bp_{k,m}$.

As mentioned above, the interference coefficients are given by the \emph{real part} of the elements of $\tilde{\bG}_{mm',nn'}$. Let $\bR_{mm',nn'} \triangleq \Re \{ \tilde{\bG}_{mm',nn'} \}$, and $\bnu_{m,n}'' \triangleq \Re \{ \bnu'_{m,n} \}$. Accordingly, (\ref{eqn:mimo_est_d1}) can be reformulated as
\begin{align} \label{eqn:mimo_est_d1_real}
\hat{\bd}_{m,n} = \sum_{n'=-\infty}^{+\infty} \sum_{m'=0}^{M-1}  \bR_{mm',nn'} \bd_{m',n'} + \bnu_{m,n}'' .
\end{align}
By stacking the real and imaginary parts of the matrices and vectors that constitute the elements of $\tilde{\bG}_{mm',nn'}$, it is possible to find an expression for the elements of $\bR_{mm',nn'}$. In particular, for an arbitrary complex matrix or vector $\ba$, we define $\check{\ba} \triangleq \left[ \Re \{\ba^{\rm T} \}, \Im \{\ba^{\rm T} \} \right]^{\rm T}$. Thus, following (\ref{eqn:G_elements2}) we can find the elements of $\bR_{mm',nn'}$ as
\be \label{eqn:GR_elements}
R_{mm',nn'}^{k,k'} = \big( \check{\bpsi}_{mm',nn'}^k \big)^{\rm T}  \check{\bgamma}_m^{k,k'} + \delta_{kk'} \delta_{mm'} \delta_{nn'}.
\ee
We note that the real-valued random vector $\check{\bgamma}^{kk'}_m$ is zero-mean and its covariance matrix can be determined using (\ref{eqn:g_mrc_stats}) as
\begin{align} 
\bC_{m}^{k,k'} &\triangleq \mathbb{E} \Big\{ \big( \check{\bgamma}^{k,k'}_m - \mathbb{E} \big\{ \check{\bgamma}^{k,k'}_m \big\} \big) \big(\check{\bgamma}^{k,k'}_m - \mathbb{E} \big\{ \check{\bgamma}^{k,k'}_m \big\} \big)^{\mathrm{T}} \Big\} \nonumber \\
&= \frac{1}{2} \begin{bmatrix}
\Re \{ \bGamma_{m}^{k,k'} + \bK_{m}^{k,k'} \}  & \Im \{- \bGamma_{m}^{k,k'} + \bK_{m}^{k,k'} \} \\ 
\Im \{ \bGamma_{m}^{k,k'} + \bK_{m}^{k,k'} \} & \Re \{ \bGamma_{m}^{k,k'} - \bK_{m}^{k,k'} \} 
\end{bmatrix} \nonumber \\
&= \frac{1}{N} \big( {\BD}_{{\rm p}_{k'}} + \delta_{kk'} \BP_{k,m} \big) , \label{eqn:C_mrc}
\end{align}
where 
$
{\BD}_{{\rm p}_{k'}} \triangleq \frac{1}{2} \begin{bmatrix}
\bD_{{\rm p}_{k'}} & \bzero \\ 
\bzero             & \bD_{{\rm p}_{k'}}
\end{bmatrix} 
$
and
$
\BP_{k,m}  \triangleq \frac{1}{2} \begin{bmatrix}
\Re \{ \bp_{k,m} \bp_{k,m}^{\rm T} \} &  \Im \{ \bp_{k,m} \bp_{k,m}^{\rm T} \} \\ 
\Im \{ \bp_{k,m} \bp_{k,m}^{\rm T} \} & -\Re \{ \bp_{k,m} \bp_{k,m}^{\rm T} \} 
\end{bmatrix} 
$ .

Following (\ref{eqn:GR_elements}), the instantaneous power corresponding to $R_{mm',nn'}^{k,k'}$ can be calculated as
\begin{align} 
P_{mm',nn'}^{k,k'} &= \left( R_{mm',nn'}^{k,k'} \right)^2  \nonumber \\
&= \big( \check{\bgamma}^{k,k'}_m \big)^{\rm T} \bPsi_{mm',nn'}^k \check{\bgamma}^{k,k'}_m + \delta_{kk'} \delta_{mm'} \delta_{nn'} \nonumber \\
& \hspace{5mm} + 2\delta_{kk'} \delta_{mm'} \delta_{nn'} \hspace{2pt} \big( \check{\bpsi}_{mm',nn'}^k \big)^{\rm T} \check{\bgamma}^{k,k'}_m , \nonumber
\end{align}
where $\bPsi_{mm',nn'}^k \triangleq \check{\bpsi}_{mm',nn'}^k \big( \check{\bpsi}_{mm',nn'}^k \big)^{\rm T}$. From the above equation, the average power, with averaging over different channel realizations, can be calculated according to \cite[p.~53]{mathai1992quadratic},
\begin{align}
&\bar{P}_{mm',nn'}^{k,k'} \nonumber \\
&\-\ = {\rm tr} \Big\{ \bC_m^{k,k'} \bPsi_{mm',nn'}^k \Big\} + \delta_{kk'} \delta_{mm'} \delta_{nn'} \nonumber \\
&\-\ = \frac{1}{N} {\rm tr} \left\{ ( \BD_{{\rm p}_{k'}} + \delta_{kk'} \BP_{k,m})  \bPsi_{mm',nn'}^k \right\} + \delta_{kk'} \delta_{mm'} \delta_{nn'} . \label{eqn:avg_power_mrc}
\end{align}
Thus, the SINR can be calculated as given in the following proposition.

\begin{figure*}[!t]
MRC:
\begin{align} \label{eqn:SINR_mrc}
{\rm SINR}_{m,n}^k = \frac{N + {\rm tr} \big\{ \big( \BD_{{\rm p}_{k}} + \BP_{k,m} \big)  \bPsi_{mm,nn}^k  \big\} }{ \sum\limits_{\substack{k'=0 \\ k' \neq k}}^{K-1} \sum\limits_{n'=-\infty}^{+\infty} \sum\limits_{m'=0}^{M-1} {\rm tr} \big\{ \BD_{{\rm p}_{k'}} \bPsi_{mm',nn'}^k  \big\} + \mathop{\sum\limits_{n'=-\infty}^{+\infty} \sum\limits_{m'=0}^{M-1}}\limits_{(m',n')\neq(m,n)}  {\rm tr} \big\{ \big( \BD_{{\rm p}_{k}} + \BP_{k,m} \big)  \bPsi_{mm',nn'}^k  \big\} + \sigma_\nu^2} 
\end{align}
ZF:
\begin{align} \label{eqn:SINR_zf}
{\rm SINR}_{m,n}^{k} = \frac{N - K }{ \sum\limits_{k'=0}^{K-1} \hspace{-0pt} \mathop{\sum\limits_{n'=-\infty}^{+\infty} \sum\limits_{m'=0}^{M-1}}\limits_{(m',n')\neq(m,n)}  {\rm tr} \big\{ \big( \BD_{{\rm p}_{k'}} - \tilde{\BP}_{k',m} \big)  \bPsi_{mm',nn'}^k  \big\} + \sigma_\nu^2} 
\end{align}
\hrulefill
\end{figure*}

\begin{proposition} \label{prp:sinr_mrc}
In the uplink of an FBMC massive MIMO system with MRC combiner and the proposed PDP equalizer, the effective SINR can be calculated according to (\ref{eqn:SINR_mrc}) on the top of the next page.
\end{proposition}
\begin{proof}
This follows from (\ref{eqn:avg_power_mrc}), and noting that 
\begin{align*}
&{\rm SINR}_{m,n}^k \\
&\triangleq\frac{\bar{P}_{mm,nn}^{k,k}}{ \mathop{\sum\limits_{n'=-\infty}^{+\infty} \sum\limits_{m'=0}^{M-1}}\limits_{(m',n')\neq(m,n)} \bar{P}_{mm',nn'}^{k,k'} + \hspace{-8pt} \sum\limits_{n'=-\infty}^{+\infty} \sum\limits_{m'=0}^{M-1} \sum\limits_{\substack{k'=0 \\ k'\neq k}}^{K-1}  \bar{P}_{mm',nn'}^{k,k'} + \sigma_\nu^2} . 
\end{align*}
\end{proof}

\subsection{ZF} \label{sec:sinr_zf}

In the Appendix, it is shown that for the ZF combiner, provided that $N \geq K+1$, the first and second order statistics of the random vector $\bg_m^{k,k'}$ can be calculated according to
\bse \label{eqn:g_zf_stats}
\begin{align} 
\bmu_m^{k,k'} &= \delta_{kk'} \bp_{k,m}, \label{eqn:g_zf_mean} \\
\bGamma_{m}^{k,k'} &= \frac{1}{N-K} \left( \bD_{{\rm p}_{k'}} - \bp_{k',m} \bp_{k',m}^{\rm H} \right) , \label{eqn:g_zf_cov} \\
\bK_{m}^{k,k'} &= \bzero . \label{eqn:g_zf_rel}
\end{align}
\ese
Hence, the covariance matrix of $\check{\bgamma}_m^{k,k'}$ is determined by
\be \label{eqn:C_zf} 
\bC_{m}^{k,k'} = \frac{1}{N-K} \big( \BD_{{\rm p}_{k'}} - \tilde{\BP}_{k',m} \big) ,
\ee
where
$
\tilde{\BP}_{k,m}  \triangleq \frac{1}{2} \begin{bmatrix}
\Re \{ \bp_{k,m} \bp_{k,m}^{\rm H} \} & - \Im \{ \bp_{k,m} \bp_{k,m}^{\rm H} \} \\ 
\Im \{ \bp_{k,m} \bp_{k,m}^{\rm H} \} &   \Re \{ \bp_{k,m} \bp_{k,m}^{\rm H} \} 
\end{bmatrix} 
$.

\begin{proposition} \label{prp:sinr_zf}
In the uplink of an FBMC massive MIMO system with ZF combiner and the proposed PDP equalizer, and provided that $N \geq K+1$, the effective SINR can be calculated according to (\ref{eqn:SINR_zf}) on the top of the next page.
\end{proposition}
\begin{proof}
This follows from the covariance matrix given in (\ref{eqn:C_zf}) and similar analysis as in the MRC case.
\end{proof}

\section{Numerical Results} \label{sec:numerical_results}

In this section, we deploy computer simulations to evaluate the efficacy of the proposed equalization method as well as the analysis of the previous sections. For all the simulations in this section, we let $M=512$ and assume $K=10$ terminals in the network. We consider the PHYDYAS prototype filter, \cite{bellanger2010fbmc}, with the overlapping factor $\kappa=4$. Normalized exponentially decaying channel PDPs $p_k[l] = e^{-\alpha_k l} / \big( \sum_{\ell=0}^{L_{\rm h}-1} e^{-\alpha_k \ell} \big), l=0\dots,L_{\rm h}-1$ for $k=0,\cdots,K-1$ with different decaying factors $\alpha_k = (k+1)/20$ for different terminals and length $L_{\rm h}=50$ are assumed\footnote{A similar approach has been taken in \cite{pitarokoilis2012optimality} to choose the channel PDPs for different terminals.}. Notice that the channel PDPs are normalized, i.e., $\sum_{l=0}^{L_{\rm h}-1} p_k[l] = 1$, for $k = 0,\dots,K-1$. Hence, following the channel model in Section \ref{sec:mmimo_fbmc}, the average signal-to-noise ratio (SNR) at the BS antennas input can be calculated as ${\rm SNR} = 1/\sigma_\nu^2$. We present the SINR performance corresponding to terminal $k=0$.

First, we show the SINR for the case where the proposed equalization is not incorporated at the BS. Fig.~\ref{fig:sinr_sat} shows the average SINR performance (with averaging over different channel realizations) of MRC, ZF, and MMSE combiners as a function of different number of BS antennas. The noise level is selected such that the SNR at the input of the BS antennas is equal to $10$~dB. From Fig.~\ref{fig:sinr_sat}, we can see that without the proposed equalization, the SINR performance of all three linear detectors, i.e., MRC, ZF, and MMSE, tend to the saturation level predicted by (\ref{eqn:SINR_sat}) as $N$ grows large. Accordingly, arbitrarily large SINR values cannot be achieved by increasing the BS array size. Also, the SINR performance of ZF and MMSE combiners converges faster to the saturation level as compared to the one in MRC. In practice, when considering a finite number of BS antennas, the impact of SINR saturation depends on the combining method used as well as the channel PDP and noise level.

\begin{figure}[!t]
\centering
\includegraphics[scale=0.6]{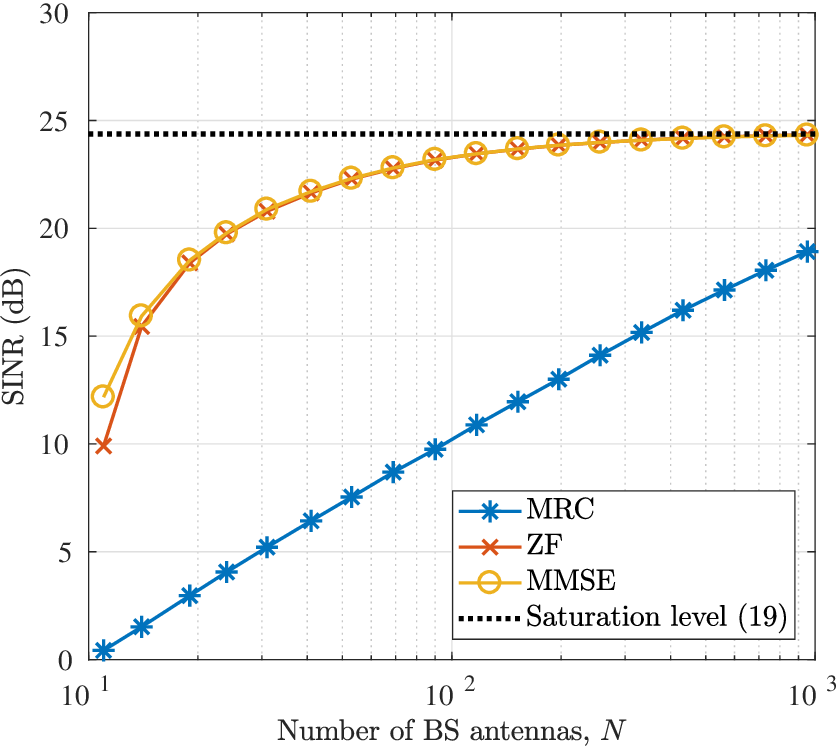}
\caption{SINR performance comparison for the case that the proposed equalizer is not utilized.} 
\label{fig:sinr_sat}
\end{figure}

\begin{figure}[!t]
\centering
\includegraphics[scale=0.6]{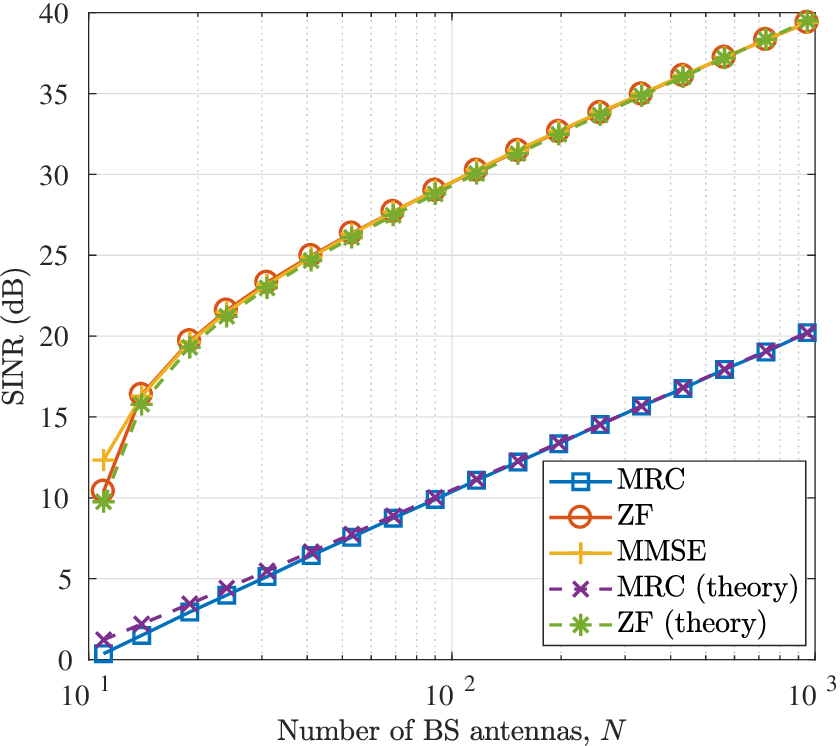}
\caption{SINR performance comparison for the case that the proposed equalizer is utilized. } 
\label{fig:sinr}
\end{figure}

In the next set of simulations, we evaluate the performance of FBMC with the proposed equalizer in place. Fig.~\ref{fig:sinr} shows the SINR performance of MRC, ZF, and MMSE combiners as a function of different number of BS antennas. 
The noise level is selected such that the SNR at the input of the BS antennas is equal to $10$~dB. As it is shown, using the proposed equalization method, the saturation problem of the conventional FBMC systems in massive MIMO channels is avoided and arbitrarily large SINR values can be achieved by increasing $N$. In Fig.~\ref{fig:sinr}, we have also shown the theoretical SINR values for MRC and ZF combiners, as calculated in (\ref{eqn:SINR_mrc}) and (\ref{eqn:SINR_zf}), respectively. This figure confirms that the theoretical SINR values match the simulated ones. This verifies the accuracy of the analysis of Section \ref{sec:sinr}.

Fig.~\ref{fig:sinr_snr} shows the theoretical SINR performance of the MRC and ZF combiners and with the proposed equalization as a function of different input SNR values. Moreover, the SINR performance of OFDM with MRC and ZF combiners is shown as a benchmark; see \cite{ngo2013energy} for the SINR expressions of OFDM. In this figure, we consider $N = 100$ BS antennas. As the figure shows, OFDM and FBMC have almost identical SINR performance when MRC is utilized. On the other hand, in the case of ZF combiner, although the performance of OFDM and FBMC are very close in the low SNR regime, a better SINR is expected for OFDM in the high SNR region. The reason for this phenomenon is that in OFDM, the interference is entirely removed using the CP. Hence, by increasing the input SNR, a better SINR at the output is also expected. In contrast, the FBMC waveform is designed to increase the bandwidth efficiency by excluding the CP overhead and providing much lower out-of-band emission than OFDM. Hence, due to the absence of CP, some residual interference remains after the ZF combining. This residual interference becomes noticeable only in the very high SNR regime.

\begin{figure}[!t]
\centering
\includegraphics[scale=0.6]{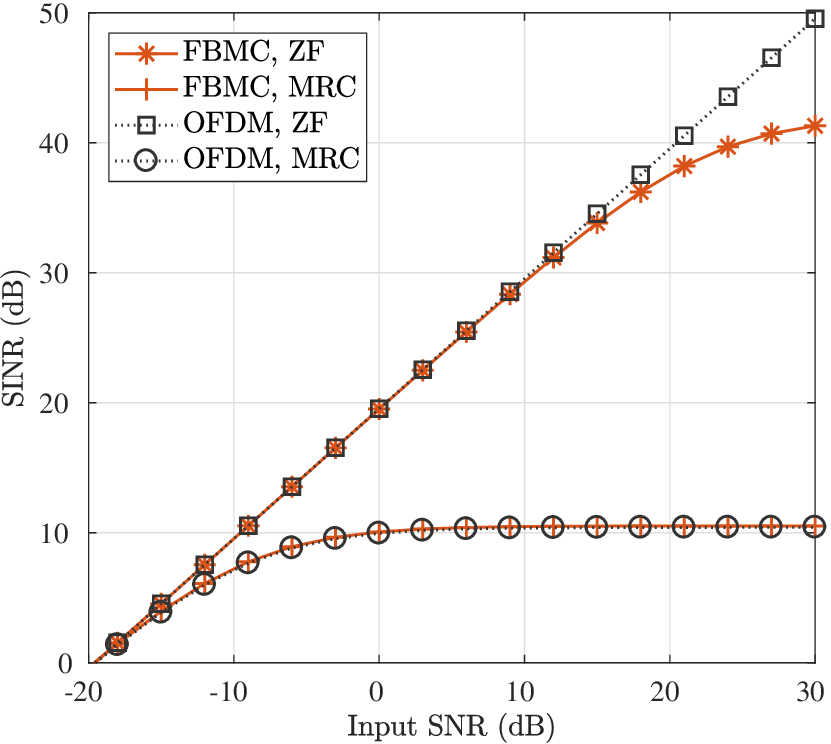}
\caption{SINR performance comparison as a function of different SNR values. In the case of FBMC, the proposed equalizer is incorporated at the BS. Here $N = 100$ BS antennas is considered.} 
\label{fig:sinr_snr}
\end{figure}

As discussed in Section \ref{sec:self-equalization}, by incorporating a large number of BS antennas, one can widen the subcarrier bands in an FBMC system. This, in turn, brings a number of advantages, e.g., robustness to CFO and channel time variations, lower PAPR, lower latency, higher bandwidth efficiency. These benefits are crucial for the next generation of wireless systems. In the next experiment, we aim at evaluating the SINR performance as we widen the subcarrier bands. Fig.~\ref{fig:sinr_symbol_spacing} shows the SINR for different values of FBMC subcarrier spacings, $\Delta F \triangleq 1/M$. In this experiment, the input SNR of $0$~dB is considered. To use the simple single-tap equalizer per subcarrier, the design norm is to choose the symbol spacing to be about an order of magnitude larger than the channel length. In this case, with $L_{\rm h}=50$, this leads to the symbol spacing of around $M/2 = 500$, which in turn yields the subcarrier spacing of $\Delta F = 0.001$. However, as the figure shows, by incorporating a large number of BS antennas as well as the proposed equalizer, one can considerably increase the subcarrier spacing while the SINR performance has a slight degradation. In particular, increasing the subcarrier spacing by an order of magnitude leads to about 0.7~dB SINR degradation when using ZF combiner. In MRC, the degradation is negligible, i.e., less than 0.3 dB.

\begin{figure}[!t]
\centering
\includegraphics[scale=0.6]{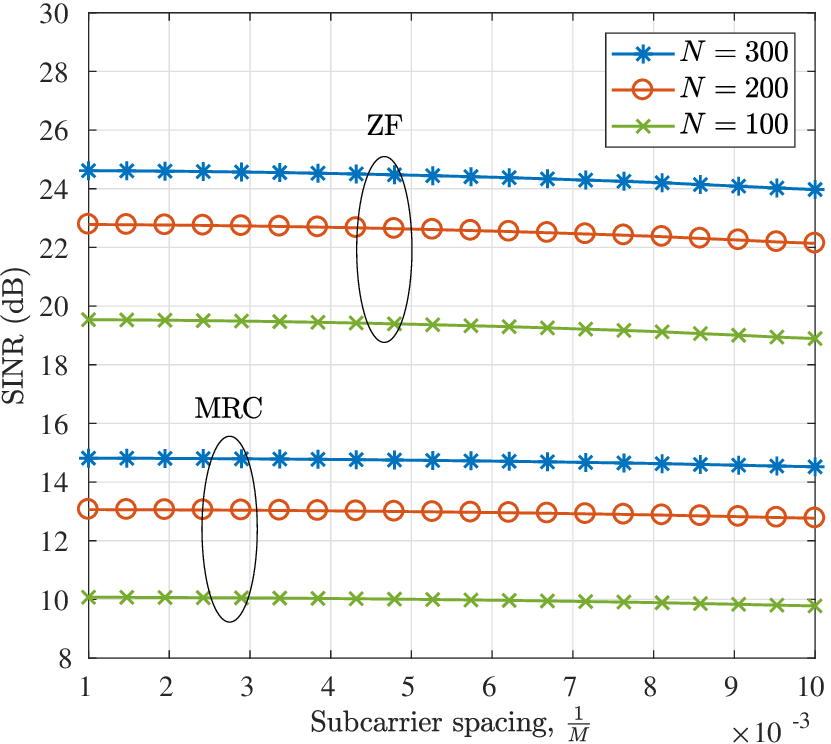}
\caption{SINR performance comparison for different values of the FBMC subcarrier spacing $\Delta F \triangleq 1/M$.} 
\label{fig:sinr_symbol_spacing}
\end{figure}

Fig.~\ref{fig:ber} presents the uncoded bit error rate (BER) performance comparison. In this experiment, $N=100$ BS antennas is considered. Moreover, the transmitted symbols belong to a 64 quadrature amplitude modulation (QAM) constellation. We compare the performance of FBMC with and without our proposed channel PDP equalizer. We also show the performance of OFDM as a benchmark. For all cases, ZF combiner is utilized. As the figure shows, the BER performance is improved significantly when the proposed channel PDP equalizer is in place. Furthermore, we achieve the same performance as in OFDM, where the channel frequency response is completely flat over each individual subcarrier band.

\begin{figure}[!t]
\centering
\includegraphics[scale=0.6]{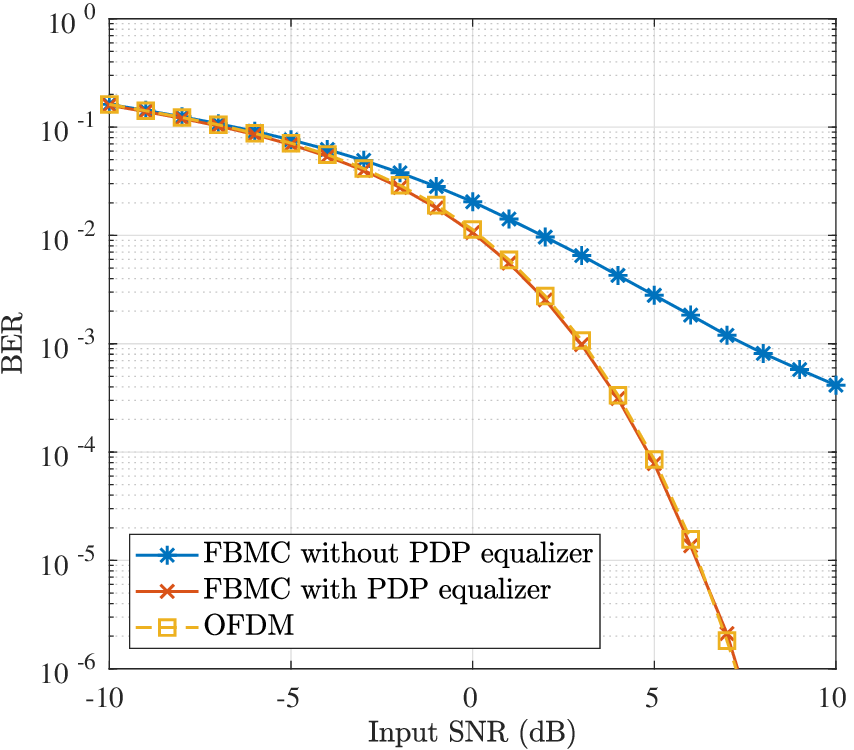}
\caption{BER performance comparison. Here, $N=100$ BS antennas and the ZF combiner are considered.} 
\label{fig:ber}
\end{figure}

\section{Conclusion and Discussion} \label{sec:conclusion}

In this paper, we studied the performance of FBMC transmission in the context of massive MIMO. We considered single-tap-per-subcarrier equalization using the conventional linear combiners, i.e., MRC, ZF, and MMSE. It was shown that the correlation between the multi-antenna combining tap weights and the channel impulse responses leads to an interference which does not fade away even with an infinite number of BS antennas. Hence, arbitrarily large SINR values cannot be achieved, and the SINR is upper-bounded by a certain deterministic value. We derived a closed-form expression for this upper bound, identified the source of SINR saturation, and proposed an efficient equalization method to remove the above correlation and resolve the problem. We mathematically analyzed the performance of the FBMC system incorporating the proposed equalization method and derived closed-form expressions for the SINR in the cases of MRC and ZF.

Throughout this paper, we assumed a \emph{co-located} BS antenna array that is sufficiently compact so that the channel responses corresponding to a particular user and different BS antennas are subject to the same channel PDP. It is worth mentioning that there exist another type of massive MIMO setup in which the elements of the BS array are distributed in a large area. This setup, which is out of the scope of this paper, is called \emph{distributed} or \emph{cell-free} massive MIMO, \cite{ngo2017cell}. In this scenario, for a given user, channel responses corresponding to different BS antennas undergo different PDPs. This is a completely different problem than what we are considering in this paper. Thus, it remains as a future study.

The analyses in this paper was based on single-tap per subcarrier equalization. However, we note that as mentioned in Section \ref{sec:mmimo_fbmc}, the asymptotic SINR saturation is an inherent property of FBMC-based massive MIMO systems due to the absence of CP. As a result, one may expect the SINR saturation issue to appear also in FBMC systems incorporating multi-tap per subcarrier equalization methods such as those in \cite{aminjavaheri2015frequency} and \cite{ihalainen2011channel} if we do not equalize the channel PDP. Using multi-tap equalizers, however, can increase the saturation level in expense of a higher computational cost. We can also realize this point from the results of \cite{aminjavaheri2015frequency}, where the performance of multi-tap and single-tap equalizers are compared with each other for different number of BS antennas. Therefore, our proposed channel PDP equalizer can also be adopted in multi-tap systems to further improve the performance.


\appendix[Proof of (\ref{eqn:g_mrc_stats}) and (\ref{eqn:g_zf_stats})]
\section{Proof of (\ref{eqn:g_mrc_stats}) and (\ref{eqn:g_zf_stats})} \label{appendix:g}

\subsection{MRC}

In the case of MRC, we have $\bg_m^{k,k'} = \frac{1}{D_m^{k,k}} \sum_{i=0}^{N-1} \left( H_{m}^{i,k} \right)^* \bh_{i,k'}$. Moreover, for large values of $N$, $D_m^{k,k}$ tends to $N$ due to the law of large numbers. Hence, the mean of the $\ell^{\rm th}$ element of $\bg_m^{k,k'}$, for $\ell \in \{0,\dots,L_{\rm h}-1\}$, can be calculated as
\begin{align*}
\BE\{g_m^{k,k'}[\ell]\} &= \frac{1}{N} \sum_{i=0}^{N-1} \sum_{l=0}^{L_{\rm h}-1} \mathbb{E}\{ h_{i,k}^*[l] h_{i,k'}[\ell] \} e^{j\frac{2\pi l m}{M}} \\
&= \delta_{kk'} p_{k,m}[\ell] .
\end{align*}
This leads to (\ref{eqn:g_mrc_mean}). We now calculate the correlation between $g_m^{k,k'}[\ell]$ and $g_m^{k,k'}[\ell']$, for $\ell,\ell' \in \{0,\dots,L_{\rm h}-1\}$. We consider the case that $k \neq k'$. Hence, 
\begin{align*}
&\BE \big\{ g_m^{k,k'}[\ell] \big( g_m^{k,k'}[\ell'] \big)^* \big\}  \\
&= \frac{1}{N^2} \sum_{i=0}^{N-1} \sum_{i'=0}^{N-1} \BE \big\{ \left(H_m^{i,k}\right)^* H_m^{i',k} h_{i,k'}[\ell] h_{i',k'}^*[\ell'] \big\} \\
&= \frac{1}{N^2} \sum_{i=0}^{N-1} \sum_{i'=0}^{N-1} \sum_{l=0}^{L_{\rm h}-1} \sum_{l'=0}^{L_{\rm h}-1} \BE \{ h_{i,k}^*[l] h_{i',k}[l'] h_{i,k'}[\ell] h_{i',k'}^*[\ell'] \} \\
& \-\ \-\ \times e^{j\frac{2\pi (l-l')m}{M}} = \frac{1}{N}  \delta_{\ell \ell'} p_{k'} [\ell] ,  \hspace{8mm} \text{for } k\neq k'.
\end{align*}
The above correlation for the case of $k=k'$ can be determined using a similar line of derivations. The result is 
\begin{equation*}
\BE \big\{ g_m^{k,k}[\ell] \big( g_m^{k,k}[\ell'] \big)^* \big\} = \frac{1}{N} \delta_{\ell \ell'} p_{k}[\ell] + p_{k,m}[\ell] p_{k,m}^*[\ell'] .
\end{equation*}
This leads to (\ref{eqn:g_mrc_cov}). Moreover, the pseudo-covariance matrix $\bK_m^{k,k'}$ in (\ref{eqn:g_mrc_rel}) can be derived using the same line of derivations as above.

\subsection{ZF}

Here, we use similar techniques as in \cite{corvaja2010sinr}. We have $g_m^{k,k'}[\ell] = \bw_{m,k}^{\rm H} \bh_{k'}[\ell]$, where $\bw_{m,k}$ is the $k^{\rm th}$ column of the combiner matrix $\BW_m$, and $\bh_{k'}[\ell]$ is an $N\times 1$ vector with its $i^{\rm th}$ element equal to $h_{i,k'}[\ell]$. In the case of ZF equalizer, we have $\BW_m = \bH_m (\bH_m^{\rm H} \bH_m)^{-1}$. Also, let $\bh_{m,k}$ denote the $k^{\rm th}$ column of $\bH_m$. Hence, the mean of $g_m^{k,k'}[\ell]$ can be determined as follows.
\begin{align*}
&\BE \{ g_m^{k,k'}[\ell] \} = \BE \{ \bw_{m,k}^{\rm H} \bh_{k'}[\ell] \} \\
&= \frac{1}{M} \sum_{m'=0}^{M-1} \BE \{ \bw_{m,k}^{\rm H} \bh_{m',k'} \} e^{j\frac{
2\pi m' \ell}{M}} \\
&\stackrel{(a)}{=} \frac{1}{M} \sum_{m'=0}^{M-1}  \sum_{l=0}^{L_{\rm h}-1} \BE \{ \bw_{m,k}^{\rm H} \bh_{m,k'} \} p_{k',m}[l] e^{j\frac{2\pi m' (\ell-l)}{M}}  \\
&\stackrel{(b)}{=} \frac{1}{M} \sum_{m'=0}^{M-1}  \sum_{l=0}^{L_{\rm h}-1} \delta_{kk'} p_{k,m}[l]  e^{j\frac{2\pi m' (\ell-l)}{M}}  = \delta_{kk'} p_{k,m}[\ell]  .
\end{align*}
This results in (\ref{eqn:g_zf_mean}). In the above equation, (a) follows from the fact the channel frequency response $\bh_{m',k'}$ can be expressed as a combination of a term that is correlated with $\bh_{m,k'}$ and a term that is independent of $\bh_{m,k'}$, i.e.,
\be \label{eqn:indep} 
\bh_{m',k'} = \alpha_{mm',k'} \bh_{m,k'} + \bh^{\rm indep}_{mm',k'}, 
\ee
where $\bh^{\rm indep}_{mm',k'}$ is independent of $\bh_{m,k'}$ and the correlation coefficient $\alpha_{mm',k'}$ can be calculated as 
\begin{equation*}
\alpha_{mm',k'} = \BE \big\{ H_{m'}^{i,k'} \big(H_{m}^{i,k'}\big)^* \big\} = P_{k'}[m'-m] ,
\end{equation*}
where $P_{k}[m] \triangleq \sum_{l=0}^{L_{\rm h}-1} p_k[l] e^{-j\frac{2\pi ml}{M}}$ is the $m^{\rm th}$ coefficient of the $M$-point discrete Fourier transform of the channel PDP $p_k[l]$. The step (b) above follows from the fact that in the case of ZF equalization, we have $\bw_{m,k}^{\rm H} \bh_{m,k'} = \delta_{kk'}$, which results from $\BW_m^{\rm H} \bH_m = \eye_k$.

In order to calculate the covariance matrix $\bGamma_m^{k,k'}$ in (\ref{eqn:g_zf_cov}), we now find the correlation between $g_m^{k,k'}[\ell]$ and $g_m^{k,k'}[\ell']$, for $\ell,\ell' \in \{0,\dots,L_{\rm h}-1\}$. We have,
\begin{flalign*}
& \BE \big\{ g_m^{k,k'}[\ell] \big( g_m^{k,k'}[\ell'] \big)^* \big\} = \BE \{ \bw_{m,k}^{\rm H} \bh_{k'}[\ell] \bh_{k'}^{\rm H}[\ell'] \bw_{m,k}  \} \\
& \stackrel{(a)}{=} \delta_{kk'} p_{k,m}[\ell] p_{k,m}^*[\ell'] + \frac{1}{M^2} \sum_{m'=0}^{M-1}\sum_{m''=0}^{M-1} \\
& ~~ \BE\{ \bw_{m,k}^{\rm H} \bh_{mm',k'}^{\rm indep}  \big( \bh_{mm'',k'}^{\rm indep} \big)^{\rm H} \bw_{m,k} \} e^{j\frac{2\pi m' \ell}{M}} e^{-j\frac{2\pi m'' \ell'}{M}} \\
& \stackrel{(b)}{=} \delta_{kk'} p_{k,m}[\ell] p_{k,m}^*[\ell'] + \frac{1}{M^2(N-K)} \sum_{m'=0}^{M-1}\sum_{m''=0}^{M-1} \\
& ~~ \Big( P_{k'}[m'-m''] - P_{k'}[m'-m] P_{k'}[m-m''] \Big) e^{j\frac{2\pi (m' \ell - m'' \ell')}{M}}  \\
& \stackrel{(c)}{=}  \delta_{kk'} p_{k,m}[\ell] p_{k,m}^*[\ell'] \hspace{-1pt} + \hspace{-1pt} \frac{1}{N-K} \big( \delta_{\ell \ell'} p_{k'}[\ell] \hspace{-1pt}  -  \hspace{-1pt} p_{k',m}[\ell] p_{k',m}^*[\ell'] \big)  . 
\end{flalign*}
This results in (\ref{eqn:g_zf_cov}). In the above equation, equality (a) follows from (\ref{eqn:indep}). Then, equality (b) follows from the independence of $\bw_{m,k}$ from $\bh_{mm',k'}^{\rm indep}$ and $\bh_{mm'',k'}^{\rm indep}$, the correlation
\begin{align*}
&\BE \Big\{ \bh_{mm',k'}^{\rm indep} \big( \bh_{mm'',k'}^{\rm indep} \big)^{\rm H} \Big\} \\
& ~~~ = \Big( P_{k'}[m'-m''] - P_{k'}[m'-m] P_{k'}[m-m''] \Big) \eye_N ,
\end{align*}
and the identity 
\begin{equation*}
\BE \Big[ {\rm tr} \big\{ \big(\BW_m^{\rm H} \BW_m\big)^{-1} \big\} \Big] =  \BE \Big[  {\rm tr} \big\{ \big(\bH_m^{\rm H} \bH_m\big)^{-1} \big\} \Big] = \frac{K}{N-K} ,
\end{equation*}
for $N \geq K+1$. The latter identity is based on the fact that $\bH_m^{\rm H} \bH_m$ is a $K \times K$ complex central Wishart matrix with $N$ degrees of freedom and covariance $\eye_K$, \cite{tulino2004random}. Finally, the equality (c) above follows using some straightforward algebraic manipulations. We note that using a similar line of derivations as above, one can find the pseudo-covariance matrix given in (\ref{eqn:g_zf_rel}).

\bibliographystyle{IEEEtran} 
\bibliography{IEEEabrv,MM-FBMC}

\end{document}